\documentclass[a4paper]{article}

\usepackage{a4wide}
\usepackage[T1]{fontenc}
\usepackage[utf8]{inputenc}
\usepackage{authblk}

\usepackage[table]{xcolor}
\definecolor{lilla}{HTML}{750787}
\usepackage{amssymb,amsmath,amsthm}
\usepackage{hyperref}
\usepackage[capitalise]{cleveref}
\usepackage{graphicx}
\usepackage{caption}
\usepackage{subcaption}
\usepackage{booktabs}
\usepackage{tabularx}
\usepackage{multicol}
\usepackage{gensymb}
\usepackage{enumerate}   
\usepackage{csquotes}
\usepackage[hyperpageref]{backref}

\renewcommand*{\backref}[1]{}
\renewcommand*{\backrefalt}[4]{\ifcase #1\or [p.~#2.]\else [pp.~#2.]\fi }

\hypersetup{
	colorlinks,citecolor=green!50!black,linkcolor=red!50!black,urlcolor=black
}

\usepackage{tikz}

\usepackage{mathtools}
\usepackage{xargs}
\newcommandx{\set}[2][1=1]{\ensuremath{\{#1,\ldots,#2\}}}
\newcommandx{\tlog}[3][1=,3=]{\log_{#1}^{#3}(#2)}
\newcommandx{\ith}[2][1=th]{#2\nobreakdash-#1}

\newtheorem{lemma}{Lemma}
\newtheorem{theorem}{Theorem}
\newtheorem{hypothesis}{Hypothesis}
\newtheorem{proposition}{Proposition}
\newtheorem{observation}{Observation}
\theoremstyle{definition}
\newtheorem{problem}{Problem}

\newtheorem{construction}{Construction}

\newcommand{\cqed}{\hfill$\diamond$}

\crefname{observation}{Observation}{Observations}
\crefname{rrule}{Reduction Rule}{Reduction Rules}
\crefname{construction}{Construction}{Constructions}
\crefname{theorem}{Theorem}{Theorems}
\Crefname{theorem}{Thm.}{Thms.}
\crefname{corollary}{Corollary}{Corollaries}
\crefname{lemma}{Lemma}{Lemmata}
\Crefname{corollary}{Cor.}{Cors.}
\crefname{proposition}{Proposition}{Propositions}
\Crefname{proposition}{Prop.}{Props.}

\newcommand{\ctw}{C_{12}}
\newcommand{\cei}{C_{8}}
\newcommand{\cp}{\tilde{P}}

\DeclareMathOperator{\interior}{int}
\DeclareMathOperator{\RNG}{RNG}
\DeclareMathOperator{\RCG}{RCG}
\DeclareMathOperator{\GAB}{GG}
\DeclareMathOperator{\DT}{DT}

\crefname{problem}{Problem}{Problems}
\Crefname{problem}{Prob.}{Probs.}

\newcommandx{\decprob}[6][3=Input,5=Question]{\begin{samepage}
  \begingroup
\begin{problem}\label{prob:#2}{\setlength{\fboxsep}{1pt}\colorbox{gray!17!white}{\textsc{#1}}}
  \nopagebreak[4]\end{problem}\nopagebreak[4]\vspace{-0.6em}
  \par\noindent\hangindent=\parindent\textbf{#3}:  #4\nopagebreak[4]
  \par\noindent\hangindent=\parindent\textbf{#5}:  #6
  \par\medskip
  \endgroup
  \end{samepage}
}

\newcommand{\N}{\mathbb{N}}
\newcommand{\Nzero}{\mathbb{N}_0}

\newcommand{\Z}{\mathbb{Z}}
\newcommand{\R}{\mathbb{R}}

\newcommand{\bigO}{\mathcal{O}}
\renewcommand{\O}{\bigO}
\newcommand{\fvs}{\varphi}
\newcommand{\ds}{\gamma}
\newcommand{\is}{\alpha}

\newcommand{\prob}[1]{{\normalfont\textsc{#1}}}
\newcommand{\vcTsc}{\prob{Vertex Cover}}
\newcommand{\vcAcr}{\prob{VC}}
\newcommand{\fvsTsc}{\prob{Feedback Vertex Set}}
\newcommand{\fvsAcr}{\prob{FVS}}
\newcommand{\tcTsc}{\prob{3-Colorability}}
\newcommand{\tcAcr}{\prob{3-Col}}
\newcommand{\isTsc}{\prob{Independent Set}}
\newcommand{\isAcr}{\prob{IS}}

\newcommand{\dsTsc}{\prob{Dominating Set}}
\newcommand{\dsAcr}{\prob{DS}}
\newcommand{\hcTsc}{\prob{Hamiltonian Cycle}}
\newcommand{\hcAcr}{\prob{HC}}

\newcommand{\cocl}[1]{\ensuremath{\operatorname{#1}}}

\newcommand{\NP}{\cocl{NP}}

\newcommand{\ETHbreaks}{the ETH fails}

\newcommand{\calC}{\mathcal{C}}

\newcommand{\emb}{\operatorname{emb}}

\newcommand{\tref}[1]{{\scriptsize(\Cref{#1})}}

\newcommand{\etal}{et~al.}
\newcommand{\ceq}{\ensuremath{\coloneqq}}
\newcommand{\eps}{\varepsilon}
\newcommand{\cball}[1]{\overline{B}_{#1}}
\newcommand{\tu}{\tilde{u}}

\DeclareMathOperator{\dist}{d}
\DeclarePairedDelimiter{\abs}{\lvert}{\rvert}

\newcommand{\wilog}{without loss of generality}
\newcommand{\Wilog}{Without loss of generality}

\include{figs/tikzpreamble}

\newcommand{\thetitle}{Most Classic Problems Remain NP-hard on Relative Neighborhood Graphs and their Relatives\footnote{This work is based on the first author's master thesis.}}

\date{}

\title{\thetitle} 

\author{Pascal Kunz\footnote{Partially supported by DFG Research Training Group 2434 ``Facets of Complexity''.}}

\author{Till Fluschnik\footnote{Supported by DFG, project TORE, NI~369/18.}}

\author{Rolf Niedermeier}

\author{Malte Renken\footnote{Supported by DFG, project MATE, NI~369/17}}

\affil{Technische Universit\"at Berlin, Algorithmics and Computational Complexity, Berlin, Germany \\ \symbol{123}p.kunz.1, till.fluschnik, rolf.niedermeier, m.renken\symbol{125}@tu-berlin.de}

\begin{document}
\maketitle

\begin{abstract}
Proximity graphs have been studied for several decades,
motivated by applications in computational geometry,
geography, data mining, and many other fields.
However, the computational complexity of classic graph problems on proximity graphs mostly remained open.
We now study
\tcTsc{},
\dsTsc{},
\fvsTsc{},
\hcTsc{},
and \isTsc{} 
on the proximity graph classes
relative neighborhood graphs,
Gabriel graphs, 
and relatively closest graphs.
We prove that all of the problems remain \NP-hard on these graphs,
except for \tcTsc{} and \hcTsc{} on relatively closest graphs,
where the former is trivial and the latter is left open.
Moreover,
for every \NP-hard case we additionally show
that no~$2^{o(n^{1/4})}$-time algorithm exists
unless~\ETHbreaks,
where~$n$ denotes the number of vertices.
\end{abstract}

\section{Introduction}
\label{sec:intro}

Proximity graphs describe the distance relationships between points in the plane
or higher-dimensional structures.
They are mostly studied in computational geometry, 
yet arise in several fields of science and engineering~\cite{Jaromczyk1992,Toussaint2014}:
most obviously in geography, less obviously in data mining \cite{Vries2016,Gyllensten2015}, 
computer vision \cite{Xu2019}, 
the design of mobile ad-hoc networks \cite{Baccelli2010, Papageorgiou2016}, 
the design of crowd-movement sensors \cite{chilipirea2015}, 
analyzing road traffic \cite{Watanabe2010}, and
describing the spread of a species of mold \cite{Adamatzky2009}. 
In this paper,
we study the computational complexity of classic NP-complete problems on three specific proximity graphs:
relative neighborhood graphs (RNGs)~\cite{Toussaint1980}, 
Gabriel graphs (GGs)~\cite{Gabriel1969}, 
and relatively closest graphs (RCGs)~\cite{Lankford1969}.
All three are subgraphs of the better-known Delaunay triangulation (DT).
For DTs, the restrictions of some classic \NP-complete graph problems have already been studied~\cite{Cimikowski1990,Dillencourt1996} and we extend this research to these three classes.

RNGs, GGs, and RCGs
are examples of empty region graphs~\cite{Cardinal2009}. 
Every pair of points is associated with a region in the plane,
their region of influence,
and is connected by an edge if there is no other point in that region
(see~\cref{fig:allthree}). 
\begin{figure}
 \centering
 \begin{tikzpicture}
  \def\sc{0.475}
  \def\xr{1*\sc}
  \def\yr{1*\sc}
  \tikzpreamble{}
  
  \newcommand{\thepoints}{\node (p1) at (0,2*\yr)[xnode]{};
    \node (p2) at (2*\xr,0*\yr)[xnode]{};
    \node (p3) at (2*\xr,3*\yr)[xnode]{};
    \node (p4) at (5*\xr,3/2*\yr)[xnode]{};
  }
  \newcommand{\theregion}[1]{\clip(#1*\xr,-0.4*\yr) rectangle (5.5*\xr,3.4*\yr);
  }
  \newcommand{\mylab}[1]{\node at (-0.3*\xr,3.5*\yr)[]{#1};}
  
  \begin{scope}
   \mylab{(a)};
   \thepoints{}
   \node at (p1.south)[anchor=north,align=right]{$p_1$};
   \node at (p2.south)[anchor=north,align=right]{$p_2$};
   \node at (p3.south)[anchor=north,align=right]{$p_3$};
   \node at (p4.south)[anchor=north,align=right]{$p_4$};
  \end{scope}
  
  \begin{scope}[xshift=7.25*\xr cm]
  \mylab{(b)};
   \thepoints{}
   \draw[xedge] (p2) to (p1) to (p3);
   \theregion{1.5}
   \draw[red, thick,dashed] (p2) circle (3.354102*\sc cm);
\draw[thick, green!50!black,dashdotted] (p4) circle (3.354102*\sc cm);
\end{scope}
  
  \begin{scope}[xshift=14.5*\xr cm]
  \mylab{(c)};
   \thepoints{}
   \draw[xedge] (p2) to (p1) to (p3);
   \draw[xedge] (p2) to (p4) to (p3);
   \theregion{-0.5}
   \draw[red, thick,dashed] (p2) circle (3*\sc cm);
   \draw[thick, blue,dotted] (p3) circle (3*\sc cm);
  \end{scope}
  
  \begin{scope}[xshift=21.75*\xr cm]
  \mylab{(d)};
   \thepoints{}
   \draw[xedge] (p2) to (p1) to (p3);
   \draw[xedge] (p2) to (p4) to (p3);
   \draw[xedge] (p2) to (p3);
   \theregion{-0.5}
\draw[magenta!70!black, thick,dashdotdotted] (2.5*\xr,1.75*\yr) circle (2.5*\sc cm);
  \end{scope}

\end{tikzpicture}
 \caption{(a) Set~$P$ of four points~$p_1=(0,2)$, $p_2=(2,0)$, $p_3=(2,3)$, and~$p_4=(5,3/2)$.
 (b) RCG on~$P$. 
 Point~$p_4$ is not adjacent with~$p_2$,
since~$p_3$ 
lies in the region of influence indicated with green (dashdotted) and red (dashed). 
(c) RNG on~$P$.
 Points~$p_2$ and~$p_3$ are not adjacent since~$p_1$ is their region of influence indicated with red (dashed) and blue (dotted). 
 (d) GG on~$P$.
 Points~$p_1$ and~$p_4$ are not adjacent since~$p_2$
 is in their region of influence indicated with magenta (dashdotted).}
 \label{fig:allthree}
\end{figure}
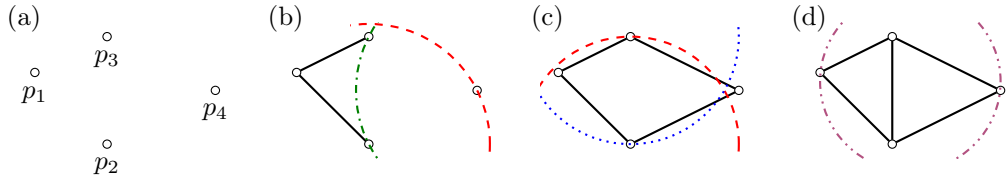
In RNGs (RCGs),
two points' region of influence is the intersection of open (closed) circles centered on each of the points with a radius equal to their distance.
In a GG, 
two points' region of influence is a circle whose center is midway between them 
and whose diameter is their distance.

\paragraph*{Motivation.}

In a railway network,
it may make sense to build a track directly
from one city to another,
if there is no third city in between.
If we interpret the area between two cities as a region of influence, then
this makes proximity graphs such as RNGs
plausible models for such networks.
One might want to build as few 
maintenance facilities for the network as possible
such that every track has a facility at one of its endpoints.
This is an instance of the \vcTsc{} (\vcAcr{}) problem
(closely related to \isTsc{} (\isAcr{})).
While \vcAcr{} is \NP-hard on general graphs,
one wonders whether it might be easier on proximity graphs. 
We will show that the problem remains \NP-hard on RCGs,
RNGs,
and GGs.

\paragraph*{Related Work.}

Existing combinatorial results on the three graph classes include listing forbidden subgraphs, permitted graphs, and bounds on the edge density~\cite{Bose2012,Cimikowski1992,Jaromczyk1992,Matula1980,Urquhart1983}.
Much algorithmic research on proximity graphs has focused on devising algorithms that efficiently compute the proximity graph from a point set (see \cite{Mitchell2017} for an overview).
On Delaunay triangulations, \hcTsc{} is \NP-hard~\cite{Dillencourt1996},
whereas \tcTsc{} is polynomial-time solvable~\cite{Cimikowski1990}.
Cimikowski conjectured \tcTsc{} to be NP-hard on RNGs and GGs~\cite{Cimikowski1989}.
Furthermore, he proposed a heuristic for coloring GGs and a linear-time algorithm for computing a~$4$-coloring in RNGs,~\cite{Cimikowski1990} but the latter has some issues, which we will discuss in \cref{sec:tc}.

\paragraph*{Our Contributions.}

\cref{tab:results} summarizes our results.
\begin{table}
 \centering
\renewcommand{\arraystretch}{1.05}
 \caption{Overview of our results.
Herein,
 $\Delta$ denotes the maximum vertex degree.
 \\
 $^\dagger$ no~$2^{o(n^{1/4})}$-time algorithm exists unless~\ETHbreaks{}, 
 where $n$ denotes the number of vertices}
 \label{tab:results}
 \begin{tabular}{@{}rp{0.18\textwidth}p{0.18\textwidth}p{0.19\textwidth}@{}}\toprule
  & RCGs & RNGs  &  GGs
  \\\midrule\midrule
  \hyperref[sec:tc]{\tcTsc{} (\tcAcr{})}  & trivial  & \multicolumn{2}{c}{\cellcolor{lightgray!15!white}{\NP-hard$^\dagger$, even if~$\Delta= 7$~~\tref{thm:tcNP}}} 
  \\
  \hyperref[sec:ds]{\dsTsc{} (\dsAcr{})} & \multicolumn{3}{c}{\cellcolor{lightgray!25!white}{\NP-hard$^\dagger$, even if~$\Delta= 4$~~\tref{thm:dsNP}}}
  \\
  \hyperref[sec:fvs]{\fvsTsc{} (\fvsAcr)} & \multicolumn{3}{c}{\cellcolor{lightgray!15!white}{\NP-hard$^\dagger$, even if~$\Delta= 4$~~\tref{thm:fvsNP}}}
  \\
  \hyperref[sec:hc]{\hcTsc{} (\hcAcr)} & \emph{open} & \multicolumn{2}{c}{\cellcolor{lightgray!25!white}{\NP-hard$^\dagger$, even if~$\Delta=4$~~\tref{thm:hcNP}}} 
  \\
  \hyperref[sec:is]{\isTsc{} (\isAcr)} & \multicolumn{3}{c}{\cellcolor{lightgray!15!white}{\NP-hard$^\dagger$, even if~$\Delta= 4$~\tref{thm:isNP}}}
  \\
  \bottomrule
 \end{tabular}
\end{table}
We prove that 
\tcTsc{} (\tcAcr{}), \dsTsc{} (\dsAcr{}), \fvsTsc{} (\fvsAcr), \hcTsc{} (\hcAcr), and \isTsc{} (\isAcr{})
are \NP-hard on RNGs and GGs,
in particular confirming the aforementioned conjecture by Cimikowski~\cite{Cimikowski1989}.
On RCGs \tcAcr{} is trivial, but
we prove that~\dsAcr{}, \fvsAcr{}, and \isAcr{}
remain \NP-hard.
All our \NP-hardness results hold true even
for graphs of fairly small maximum degree 
(at most seven in the case of \tcAcr{},
and four in all other cases).
We complement each \NP-hardness result
with a running-time lower bound of~$2^{o(n^{1/4})}$
based on the Exponential-Time Hypothesis,
where~$n$ is the number of vertices.

\paragraph*{Our Technique.}

In our NP-hardness proofs (see~\cref{tab:results}),
we give polynomial-time many-one reductions 
from each problem's restriction
to planar graphs with maximum degree three or four.
We proceed as follows 
(see \cref{fig:exgraph} for an illustration).
We exploit the fact that for any planar graph with maximum degree at most four,
we can compute in polynomial~time a $2$-page book embedding~\cite{Bekos2016},
a very structured representation of the input graph.
Then,
we translate the book embedding's structure into a grid-like structure.
Each reduction uses three types of gadgets: 
to represent vertices, 
to represent edges, 
and to fill the space between them
in order to prevent the appearance of unwanted edges between the other gadgets.

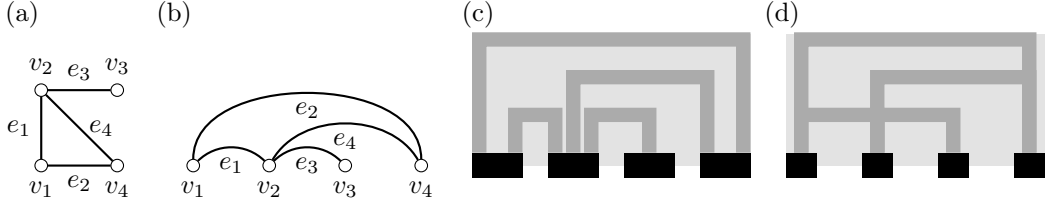
\begin{figure}
 \centering
 \begin{tikzpicture}
  \def\xr{1}
  \def\yr{1}
  \tikzpreamble{}
  \def\nsc{0.5}
  
  \begin{scope}
   \node at (-0.25*\xr,2*\yr)[]{(a)};
   \node (v1) at (0,0)[xnode,label=-90:{$v_1$}]{};
   \node (v2) at (0,1*\yr)[xnode,label=90:{$v_2$}]{};
   \node (v3) at (1*\xr,1*\yr)[xnode,label=90:{$v_3$}]{};
   \node (v4) at (1*\xr,0)[xnode,label=-90:{$v_4$}]{};
   \draw[xedge] (v1) to node[midway,left]{$e_1$}(v2);
   \draw[xedge] (v1) to node[midway,below]{$e_2$}(v4);
   \draw[xedge] (v2) to node[midway,above]{$e_3$}(v3);
   \draw[xedge] (v2) to node[midway,right]{$e_4$}(v4);
  \end{scope}
  
  \begin{scope}[xshift=2*\xr cm]
  \node at (-0.25*\xr,2*\yr)[]{(b)};
   \theexgraph{}
  \end{scope}
  
  \tikzstyle{znode}=[fill=black,minimum width=\xr*0.66 cm,minimum height=\yr*0.33 cm];
  \tikzstyle{zpath}=[-,color=gray!66!white,opacity=1,line width=0.19*\xr cm];
  \newcommandx{\theconnect}[7][5=-0.1,7=0.1]{\draw[zpath] ($(#1.north#4)+(#5*\xr,0)$) to ($(#1.north#4)+(#5*\xr,0)+(0,#3*\yr)$) to ($(#2.north#6)+(#7*\xr,0)+(0,#3*\yr)$) to ($(#2.north#6)+(#7*\xr,0)$);
  }
  \newcommand{\theconnectC}[3]{\draw[zpath] (#1.north) to ($(#1.north)+(0,#3*\yr)$) to ($(#2.north)+(0,#3*\yr)$) to (#2.north);
  }
  \begin{scope}[xshift=6*\xr cm]
   \node at (-0.25*\xr,2*\yr)[]{(c)}; 
   \draw[fill=gray!22!white,draw=none] (-0.33*\xr,-0*\yr) rectangle (3.33*\xr,1.75*\yr);
   \node (v1) at (0,0)[znode]{};
   \node (v2) at (1*\xr,0)[znode]{};
   \theconnect{v1}{v2}{0.5}{ east}{ west}
   \node (v3) at (2*\xr,0)[znode]{};
   \theconnect{v2}{v3}{0.5}{ east}{}[0]
   \node (v4) at (3*\xr,0)[znode]{};
   \theconnect{v2}{v4}{1}{}[0]{ west}
   \theconnect{v4}{v1}{1.5}{ east}{ west}
  \end{scope}
  
  \tikzstyle{znode}=[fill=black,minimum width=\xr*0.4 cm,minimum height=\yr*0.33 cm];
  \begin{scope}[xshift=10*\xr cm]
   \node at (-0.25*\xr,2*\yr)[]{(d)};
   \draw[fill=gray!22!white,draw=none] (-0.2*\xr,-0*\yr) rectangle (3.2*\xr,1.75*\yr);
   \node (v1) at (0,0)[znode]{};
   \node (v2) at (1*\xr,0)[znode]{};
   \theconnect{v1}{v2}{0.5}{}[0]{}[0]
   \node (v3) at (2*\xr,0)[znode]{};
   \theconnect{v2}{v3}{0.5}{}[0]{}[0]
   \node (v4) at (3*\xr,0)[znode]{};
   \theconnect{v2}{v4}{1}{}[0]{}[0]
   \theconnect{v4}{v1}{1.5}{}[0]{}[0]
  \end{scope}
  
 \end{tikzpicture}
 \caption{(a) A graph~$G$ with vertex set~$\{v_1,\dots,v_4\}$. 
 (b) A 1-page book embedding of~$G$.
 (c) and (d) Illustration of our technique,
 where black rectangles correspond to vertex gadgets,
 thick gray lines to edge/connector gadgets,
 and light gray areas indicate filler gadgets.}
 \label{fig:exgraph}
\end{figure}

\section{Preliminaries}
\label{sec:prelim}

Let~$\N \coloneqq \{1, 2, 3, \dots\}$ and $\Nzero \coloneqq \{0\} \cup \N$.
We use basic notions from graph theory~\cite{Diestel2017}.

\paragraph*{Proximity graphs.}
Let~$\dist \colon \mathbb{R}^2 \times \mathbb{R}^2 \rightarrow \mathbb{R}$ denote the Euclidean distance between two points.
The \emph{open} and \emph{closed ball} with radius $r>0$ and center $p\in \R^2$ are~$B_r(p)\coloneqq \{q \in \R^2 \mid \dist(p,q) < r\}$ and~$\cball{r}(p) \coloneqq \{ q \in \mathbb{R}^2 \mid \dist(p,q) \leq r\}$.
The Delaunay triangulation of $P\subseteq \R^2$ (see, e.g.,~\cite[Ch.~9]{DeBerg2008}) is denoted by $\DT(P)$.

A \emph{template region}~\cite{Cardinal2009}
is a function~$R\colon \binom{\R^2}{2} \rightarrow 2^{\R^2}$ that
assigns a region of the plane, called the \emph{region of influence},
to each pair of points in the plane.
Given a template region~$R_{\calC}$ and points $p_1,p_2 \in \R^2$, a third point $p_3 \in \R^2 \setminus \{p_1,p_2\}$ is a \emph{$\calC$-blocker} for $\{p_1,p_2\}$ if~$p_3 \in R_{\calC}(p_1,p_2)$.
For a finite set of points $P = \{p_1,\ldots, p_n\} \subseteq \R^2$, 
the $\calC$-graph \emph{induced} by $P$ is $\calC(P) \coloneqq (\{v_1,\ldots,v_n\},E_{\calC(P)})$ with $E_{\calC(P)}\coloneqq \{\{v_i,v_j\} \mid (R_{\calC}(p_i,p_j) \cap P) \setminus \{p_i,p_j\} = \emptyset \text{ and } i\neq j\}$.
The class $\calC$ contains a graph~$G$ if there is a finite set of points $P\subseteq \R^2$ with~$\calC(P) = G$.

We are interested in three template regions and the graph classes defined by them:
\begin{description}
	\item[\emph{Relatively closest graphs (RCGs)}:] Defined by
	\begin{align*}
		R_{\RCG}(p_1,p_2) & \coloneqq \{ p_3\in \R^2 \mid \dist(p_1,p_2) \geq \max\{\dist(p_1,p_3),\dist(p_2,p_3)\}\} \\
		& = \cball{\dist(p_1,p_2)}(p_1) \cap \cball{\dist(p_1,p_2)}(p_2).
	\end{align*}
	\item[\emph{Relative neighborhood graphs (RNGs)}:] Defined by
	\begin{align*}
		R_{\RNG}(p_1,p_2) & \coloneqq \{ p_3 \in \R^2 \mid \dist(p_1,p_2) > \max\{\dist(p_1,p_3),\dist(p_2,p_3)\}\}\\
		&= B_{\dist(p_1,p_2)}(p_1) \cap B_{\dist(p_1,p_2)}(p_2).
	\end{align*}
	\item[\emph{Gabriel graphs (GGs)}:] Defined by
	\begin{align*}
		R_{\GAB}(p_1,p_2) & \coloneqq \{p_3 \in \R^2 \mid \dist(p_1,p_2)^2 \geq \dist(p_1,p_3)^2 + \dist(p_2,p_3)^2\}\\
		 &=\cball{\dist (p_1,p_2)/2}(q),
	\end{align*}
	where $q$ is the midpoint between $p_1$ and $p_2$.
\end{description}
A \emph{$\calC$-embedding} of a graph~$G=(V,E)$ 
is a map~$\emb\colon V \rightarrow \R^2$ such that $\calC(\emb(V)) = G$.\footnote{To simplify notation, 
we write~$v$ instead of~$\emb(v)$ to refer to the point 
at which vertex~$v$ is embedded.}
Of course, $G\in \calC$ if and only if $G$~admits a~$\calC$-embedding.

For any finite point set $P$, 
it holds that
$E_{\RCG(P)} \subseteq E_{\RNG(P)} \subseteq E_{\GAB(P)} \subseteq \DT(P)$~\cite{Cimikowski1992}.
RCGs cannot contain~$C_3$ as a subgraph and none of the three can contain $K_4$ or $K_{2,3}$~\cite{Cimikowski1992,Matula1980,Urquhart1983}.
Moreover,~$p_3 \in R_{\GAB}(p_1,p_2)$ if and only if the angle at $p_3$ formed by the lines to $p_1$ and $p_2$ is at least~$90\degree$~\cite{Matula1980}.
Finally, the following lemma will be used to prove that graphs are in each graph class:
\begin{lemma}[\cite{Matula1980}]
	\label{lemma:planar}
	Let $P$ be a set of points in the plane and $G$ the RCG, RNG or GG induced by $P$.
	Then, the straight-line drawing of~$G$ induced by~$P$ is planar.
\end{lemma}

\paragraph*{Book embeddings.}
 Our \NP-hardness proofs use 2-page book embeddings.
 A \emph{$k$-page book embedding} of a graph~$G=(V,E)$ consists of
(i) an edge partition~$E=E_1\uplus\dots\uplus E_k$, and
(ii) for every~$i\in\set{k}$, a planar embedding~$\emb_i$ of~$(V,E_i)$ in~$\R\times\R_{\geq 0}$, where~$\emb_i(v)=\emb_j(v)\in\R\times\{0\}$ for every~$v\in V$, $i,j\in\set{k}$.
The following result due to Bekos \etal~\cite{Bekos2016} will play an important role in this work:

\begin{theorem}[\cite{Bekos2016}]
	\label{thm:prelim-compute-be}
	Every planar graph with maximum degree at most four admits a 2-page book embedding.
	Such an embedding can be computed in quadratic time.
\end{theorem}
\noindent

The following terminology will be useful in our NP-hardness proofs.
Consider a graph~$G=(V,E)$ and a 2-page book embedding of~$G$.
Let~$v_1,\ldots,v_n$ be the vertices of the graph ordered in such a way that $\emb_i(v_j) < \emb_i(v_{j+1})$ for $i\in\{1,2\}$ and every $j\in\{1,\ldots,n-1\}$.
We will say that~$v_1,\ldots,v_n$ is the \emph{order in which the vertices appear on the spine}.
Let~$r\in\{1,2\}$.
We will use~$N_r(v)\coloneqq\{v'\,|\,\{v,v'\}\in E_r\}$ 
to denote~$v$'s~\textit{$E_r$-neighborhood}
and~$\deg_r(v)\coloneqq|N_r(v)|$ to denote~$v$'s~\textit{$E_r$-degree}. 
For an edge~$e=\{v_i,v_j\}$,~$i<j$, 
define its \textit{length} as $\ell(e)\coloneqq j-i$.
The \textit{interior} of~$e\in E_r$ is
\begin{align*}
	\interior(e)\coloneqq\{e'=\{v_{i'},v_{j'}\} \in E_r \mid i\leq i'<j'\leq j, e'\neq e\}.
\end{align*}
The~\textit{$E_r$-height} of a vertex~$v$ is
\begin{align*}
	h_r(v) \coloneqq
		\max \{0, h(e)\mid v \in e \in E_r\},
		\quad\text{where}\quad
		h(e)\coloneqq
		1+ \max\{0, h(e') \mid e' \in \interior(e) \}
\end{align*}
denotes the \textit{height} of~$e$.
Let~$h_r(G)\coloneqq\max\{h_r(v_i)\,|\,i \in \{1,\ldots,n\}\}$.
Note that, because the height of any edge only depends on the height of shorter edges, edge height is well-defined.
The length and height of an edge are both in~$\bigO(n)$.
For every vertex~$v_i$, we order its incident edges in~$E_r$ as follows.
If~$N_r(v_i) = \{v_{j_1},\ldots,v_{j_k}\}$ with~$j_1<\ldots<j_c<i<j_{c+1}<\ldots<j_k$,
then the order of the edges
is~$\{v_i,v_{j_c}\}<\ldots<\{v_i,v_{j_1}\}<\{v_i,v_{j_k}\}<\ldots<\{v_i,v_{j_{c+1}}\}$.

	\paragraph{Grid structure.}
	The graphs we build in our reductions will all have a grid-like structure.
	With the exception of the reduction for \hcTsc,
	we group their vertices into \emph{$(x, y)$-corners} with $(x, y) \in \Z^2$ and there will be a corner for every $(x,y)$ within certain bounds, which depend on the problem in question as well as the size and structure of the input graph.
	In the embedding, the vertices forming the $(x,y)$-corner will be in $\cball{r}(x,y)$ for a suitable $r>0$.
	Some vertices are not part of any corner and are called \emph{intermediate vertices}.
	They are usually located midway between two corners.
	Each corner can have one or multiple dedicated right, top, left, and bottom \emph{connecting vertices}.
	If a corner consists of a single vertex, that vertex always simultaneously acts as the right, top, left and bottom connecting vertex of that corner.
	The connecting vertices of a corner are the only ones that may have neighbors outside of that corner.
For any $(x,y)\in\Z^2$,
	we say that the vertices in the $(x,y)$,~$(x+1,y)$,~$(x,y+1)$, and~$(x+1,y+1)$-corners
	along with any intermediate vertices that are adjacent to vertices in two of the aforementioned corners
	jointly form a \emph{grid face}.

\paragraph*{Exponential-Time Hypothesis.}

    \begin{hypothesis}[Exponential-Time Hypothesis (ETH)~\cite{Impagliazzo2001}]
    \label{hyp:eth}
    There is some fixed $c>0$ such that \prob{3-CNF-Sat} is not solvable in~$2^{cn} \cdot(n+m)^{\O(1)}$ time, 
    where $n$ and $m$ denote the numbers of variables and clauses, respectively.
    \end{hypothesis}

The ETH implies lower bounds for numerous problems.
We start by deriving one for \vcTsc{}, which is defined as:
\decprob{\vcTsc{} (\vcAcr{})}{vc}
{An undirected graph~$G=(V,E)$ and $k\in \Nzero$.}
{Is there a vertex set $X\subseteq V$ with $\abs{X} \leq k$ such that $G - X$ is edgeless?}
\noindent
We prove an ETH-based lower bound for the following restriction of \vcTsc.
\begin{lemma}
    \label{lemma:vc-ETH}
   \vcTsc{} restricted to planar graphs with maximum degree three admits no $2^{o(n^{1/2})}$-time
   algorithm where $n$ is number of vertices, unless~\ETHbreaks.
\end{lemma}

	\begin{proof}
		Unless~\ETHbreaks, 
		\textsc{Vertex Cover} on planar graphs admits no $2^{o(n^{1/2})}$-time algorithm where $n$ is the number of vertices~\cite{Lokshtanov2011}.
		There is a reduction~\cite{Garey1976} from \textsc{Vertex Cover} on planar graphs to the restriction of the same problem to planar graphs with maximum degree~three.
		We present a slightly modified version of this reduction which only changes the number of vertices linearly, 
		thereby proving the claim.
		
		Let $G=(V,E)$ be a planar graph and $k\in\N$.
		We construct a planar graph~$G'$ with maximum degree three 
		and an integer~$k'\in\N$ as follows. 
		For every~$v \in V$,
		let~$d\coloneqq\deg_G(v)$ and $w_1,\ldots,w_d$ be the neighbors of~$v$ in the cyclic order in which the corresponding edges enter~$v$.
		We replace~$v$ with a cycle consisting of
		the vertices~$u_1^v,\ldots,u^v_{2d}$ in that order 
		and add an additional vertex $z^v$.
		In addition to the edges in the cycle, 
		we add the edges~$\{z^v,u^v_1\}$ and~$\{u^v_{2i},w_i\}$ for every~$i \in\set{d}$.
		We set~$k'\coloneqq k + \sum_{v\in V}\deg_G(v) = k + 2|E|$.
		Note that the maximum degree in~$G'$ is three,
		that the planarity of~$G'$ follows from the planarity of~$G$,
		and that both $G'$ and~$k'$ can be computed in polynomial time.
		
		We must show that $G$~contains a vertex cover of size~$k$ 
		if and only if $G'$~contains a vertex cover size~$k'$. 
		If $X\subseteq V$ is a vertex cover of size at most~$k$, 
		then
		\begin{align*}
			X'\coloneqq
			\{u^v_1\mid v \in V\} & \cup 
			\{u^v_{2i+1} \mid i \in \{1,\ldots,\deg_G(v)-1\}, v \in V\setminus X\} \\
			& \cup \{u^v_{2i} \mid i \in \{1,\ldots,\deg_G(v)\}, v \in X \}  
		\end{align*}
		is a vertex cover of size
		\begin{align*}
            |V| + \sum_{v \in V \setminus X} (\deg_G(v) -1) + \sum_{v \in X} \deg_G(v) 
			&= |V| - |V\setminus X| + \sum_{v \in V}\deg_G(v) 
			\\
			&= |X| + \sum_{v \in V} \deg_G(v) \leq k'.
		\end{align*}
	
		Conversely, suppose that $X'$ is a vertex cover of size $k'$ in $G'$.
		Then, 
		let~$X \coloneqq \{v \in V \mid u^v_{2i} \in X \text{ for any } i \in \{1,\ldots,\deg_G(v)\}\}$.
		We claim that $X$~is a vertex cover in~$G$.
		If~$\{v,w\}\in E$,
		then~$u_{2i}^v$ or~$u_{2i'}^w$ for some~$i,i'$ must be in~$X'$ and, hence, $v$ or $w$ is in $X$.
		It remains to show that~$|X|\leq k$.
		We may assume that $u_1^v \in X'$ for all~$v \in V$, 
		since the edge $\{u_1^v,z^v\}$ must be covered.
		For every~$v \in V$, 
		let $X'_v\coloneqq X' \cap \{u^v_1,\ldots,u^v_{2\deg_G(v)}\}$.
		Because every edge in the cycle representing~$v$ is covered, 
		it follows that $|X'_v| \geq \deg_G(v)$.
		Since $\sum_{v\in V}\abs{X'_v} = \abs{\bigcup_{v\in V} X'_v} \leq k' = k + \sum_{v\in V}\deg_G(v)$, there are at most $k$ vertices with $\abs{X'_v} > \deg_G(v)$.
		Since $u_1^v \in X'_v$ holds,
		$\abs{X'_v} = \deg_G(v)$ 
		implies that $X'_v = \{u_1^v,u_3^v,\ldots,u_{2\deg_G(v)-1}^v\}$.
		Hence, 
		$\abs{X} \leq k$.
		
		The output graph $G'$ contains
		\begin{align*}
			\sum_{v \in V} (2\deg_G(v)+1) = |V| + 2\sum_{v \in V} \deg_G(v) = |V| + 4|E| \in \O(|V|)
		\end{align*}
		vertices.
		Thus,
		if \prob{Vertex Cover} on planar graphs with maximum degree three is solvable in~$2^{o(n^{1/2})}$~time,
		then \prob{Vertex Cover} on arbitrary planar graphs is solvable in~$2^{o(n^{1/2})}$~time.
\end{proof}

\section{3-Colorability}
\label{sec:tc}

In this section we study the following problem on RCGs, RNGs, and GGs.
\decprob{\tcTsc{} (\tcAcr{})}{tc}
{An undirected graph~$G=(V,E)$.}
{Is there a function $c\colon V \rightarrow \{1,2,3\}$ such that $c(u)\neq c(v)$ for all $\{u,v\} \in E$?}
\noindent
Every RCG is 3-colorable~\cite{Cimikowski1992}
because RCGs do not contain any~$3$-cycles and every planar graph without~$3$-cycles is~$3$-colorable by Grötzsch's theorem~\cite{Grunbaum1963}.
As a result, \tcAcr{} is trivial when restricted to RCGs.
Regarding RNGs and GGs, 
we prove the following.
\begin{theorem}
 \label{thm:tcNP}
 \tcTsc{} on RNGs and on GGs is \NP-hard,
 even if the maximum degree is seven.
 It admits no $2^{o(n^{1/4})}$-time algorithm
 where $n$ is the number of vertices
 unless~\ETHbreaks.
\end{theorem}
This confirms a conjecture by Cimikowski~\cite{Cimikowski1989}.
Our proof is based on a polynomial-time many-one reduction from
the \NP-hard~\cite{Garey1976c} \tcTsc{} on planar graphs with maximum degree four.

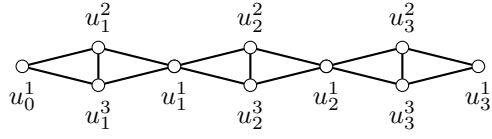
\begin{figure}[t]
	\centering
	\begin{tikzpicture}
	    \def\xr{1}
	    \def\yr{1}
	    \tikzpreamble{}
	    \def\nsc{0.5}
	    
	    \newcommand{\kfme}[6]{
            \node (a) at (0+#1*\xr,0+#2*\yr)[xnode,label=-90:{#3}]{};
            \node (b) at (1*\xr+#1*\xr,0.25*\yr+#2*\yr)[xnode,label=90:{#4}]{};
            \node (c) at (1*\xr+#1*\xr,-0.25*\yr+#2*\yr)[xnode,label=-90:{#5}]{};
            \node (d) at (2*\xr+#1*\xr,0+#2*\yr)[xnode,label=-90:{#6}]{};
            \draw[-,thick] (b) to (a) to (c) to (d) to (b) to (c);
	    }
	    
	    \kfme{0}{0}{$u^1_0$}{$u^2_1$}{$u^3_1$}{$u^1_1$};
	    \kfme{2}{0}{}{$u^2_2$}{$u^3_2$}{$u^1_2$};
	    \kfme{4}{0}{}{$u^2_3$}{$u^3_3$}{$u^1_3$};
\end{tikzpicture}
	\caption{A coloring path of length~$3$.}
	\label{fig:tc-cp}
\end{figure}
We will use so-called \emph{coloring paths} (see~\cref{fig:tc-cp} for an illustration),
which essentially allows us to copy the color of a vertex.
The \textit{coloring path} of length~$k$ from~$u^1_0$ to~$u^1_k$ is the graph~$\cp_k\coloneqq(V_k,E_k)$ with:
\begin{align*}
	V_k &\coloneqq \{u^1_0\} \cup \{u^1_i,u^2_i,u^3_i\,|\,i\in\{1,\ldots,k\}\} \text{ and }\\
	E_k &\coloneqq \{\{u^1_{i-1},u^2_i\},\{u^1_{i-1},u^3_i\},\{u^2_i,u^3_i\},\{u^2_i,u^1_{i}\},\{u^3_i,u^1_{i}\}\,|\,i\in\{1,\ldots,k\}\}.
\end{align*}
We will call~$u^1_i$ the~$i$-th \textit{center vertex},~$u^2_i$ the~$i$-th \textit{left vertex}, and~$u^3_i$ the~$i$-th \textit{right vertex}.

\begin{figure}[t]
	\centering
	\includegraphics{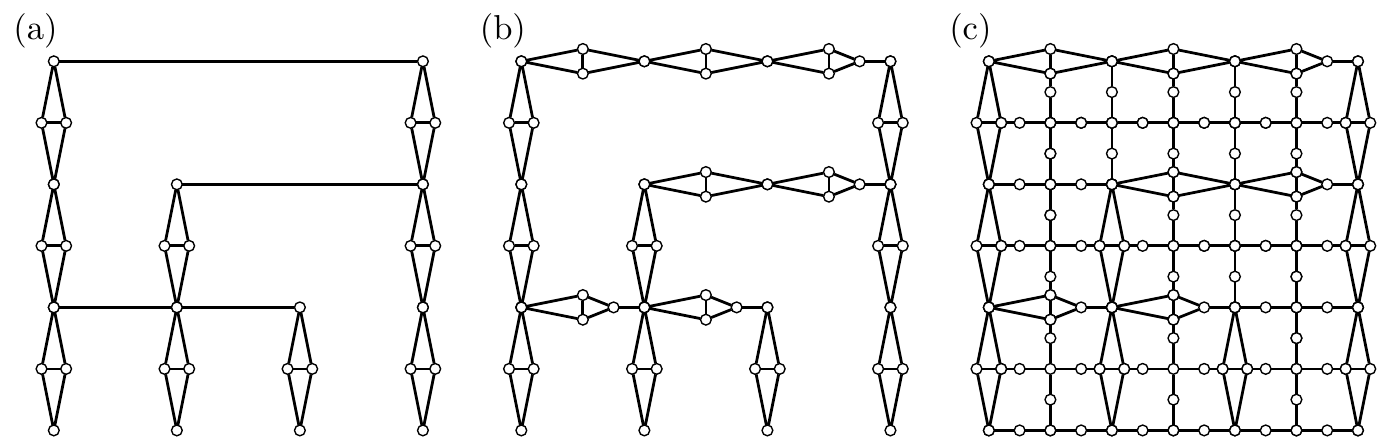}
	\caption{\cref{constr:tc}
		applied to the example graph in \cref{fig:exgraph}.
		(a)~The graph after steps 1--2. 
		(b)~The graph after steps 1--3.
		(c)~The final graph.}
	\label{fig:tc}
\end{figure}

\begin{construction}
 \label{constr:tc}
 Let~$G=(V,E)$ be an undirected planar graph of maximum degree four.
 We will construct a graph $G'=(V',E')$
 and subsequently show that $G'$~is both an RNG and a GG and that~$G$ is $3$-colorable if and only if $G'$ is (see \cref{fig:tc} for an illustration).
 The vertex set of~$G'$ will mostly consist of groups of vertices called $(x,y)$-\emph{corners}
 where $2\leq x \leq 2n$ and $-2h_2(G) \leq y \leq 2h_1(G)$.
 Each corner consists of either a single vertex or of a pair of adjacent vertices.
 Corners can have dedicated top, left, right, and bottom \emph{connecting vertices}, some of which may coincide.
 For example, if a corner consists of a single vertex, that vertex simultaneously forms all four connecting vertices.
 Finally, there will be some \emph{intermediate vertices} that are not part of any corner. We start with $G'\coloneqq G$.
 
 \textbf{Step 1:}
 Compute a 2-page book embedding of~$G$. 
 Let~$v_1,\ldots,v_n$
 be the vertices of~$G$ enumerated in the order in which they appear on the spine,
 and let~$\{E_1,E_2\}$ denote the partition of~$E$.
 
 \textbf{Step 2:} Replace every vertex $v_i$ with a coloring path of length $h_1(v_i) + h_2(v_i)$.
 Every edge~$e$ of~$G$ incident to~$v_i$ is now instead attached to the $(h_2(v_i)+h(e))$-th center vertex, if~$e \in E_1$, or the $(h_2(v_i)-h(e))$ if~$e \in E_2$.
 For $r=0,\ldots,h_1(v_i) + h_2(v_i)$, the $r$-th center vertex of that path forms the $(2i, 2r - 2h_2(v_i))$-corner.
 For $r=1,\ldots,h_1(v_i)$, the $r$-th left and right vertices jointly form the $(2i,2r-1 - 2h_2(v_i))$-corner.
 The left vertex is the left connecting vertex of this corner and the right vertex is the right connecting vertex.

 \textbf{Step 3:} For every edge $e=\{v_i,v_j\} \in E_1$, $i<j$, replace the corresponding edge of $G'$ with a coloring path of length $\ell(e)$.
 Identify the first vertex of that path with the $(2i, 2h(e))$-corner (which consists of a single vertex).
 Denote the last vertex of that path by~$w$.
 Add an edge from $w$ to the $(2j, 2h(e))$-corner.
 For $r=1,\ldots,\ell(e)-1$, the $r$-th center vertex of that path is the $(2i+2r,2h(e))$-corner.
 The vertex $w$, which is the~$\ell(e)$-th center vertex, is an intermediate vertex.
 For $r=1,\ldots,\ell(e)$, the $r$-th left and right vertices jointly form the $(2i+2r-1,2h(e))$-corner.
 The left vertex is the top connecting vertex of this corner and the right vertex is the bottom connecting vertex.
 
 \textbf{Step 4:} For every $(x,y)$ with $2\leq x \leq 2n$ and $0 \leq y \leq 2h_1(G)$,
 if an $(x,y)$-corner was not added in one of the previous two steps, then add a single vertex, which becomes the~$(x,y)$-corner, to $G'$.
 In that case add an edge from the $(x,y)$-corner to the top connecting vertex of the $(x,y-1)$-corner, to the left connecting vertex of the $(x+1,y)$-corner, and so on.
 Subdivide\footnote{If $G=(V,E)$ is a graph and $e =\{u,v\} \in E$, then the graph obtained by subdividing $e$ a total of $k$ times is the graph $G'=(V\uplus \{w_1,\ldots,w_k\},E')$ with $E'= (E \setminus \{\{u,v\}\}) \cup \{\{u,w_1\},\{w_1,w_2\},\ldots,\{w_{k-1},w_k\},\{w_k,v\}\}$.} each of these edges once, introducing four new intermediate vertices.
 
 Note that steps~3~and~4 take only~$E_1$ into account.
 These steps must be repeated analogously for~$E_2$, using negative $y$-coordinates.
 \cqed
\end{construction}

\begin{lemma}
	Let $G=(V,E)$ and $G'$ be graphs.
	\begin{enumerate}[(i)]
		\item If $G'$ is obtained from $G$ by replacing a vertex~$v \in V$ 
		with a coloring path of any length and connecting each~$u \in N_G(v)$
		with one of the center vertices of that coloring path, then $G$ 
		is~$3$-colorable if and only if $G'$ is.
		\item If $G'$ is obtained from $G$ by replacing an edge $e=\{u,v\} \in E$ with a coloring path of any length,
		identifying the first vertex of that path with~$u$ and connecting the last vertex of that path to~$v$,
		then $G$ is $3$-colorable if and only if $G'$ is.
		\item If $v_1,\ldots,v_k \in V$ and~$G'$ is obtained from~$G$ by adding~$k+1$ vertices~$w,u_1,\ldots,u_k$ and the edges~$\{v_i,u_i\}$ and~$\{u_i,w\}$ for every~$i=1,\ldots,k$, then,~$G$ is~$3$-colorable if and only if~$G'$ is.
		\item If $G'$ is obtained from $G$ by applying \cref{constr:tc} and $n$ is the number of vertices in $G$, then $G'$ contains $\bigO(n^2)$ vertices.
	\end{enumerate}
	\label{lemma:tc}
\end{lemma}

\begin{proof}
	\begin{enumerate}[(i)]
		\item Let~$\cp_k=(V_k,E_k)$ be the coloring path that replaces~$v$. 
		
		Suppose that~$c \colon V \rightarrow \{1,2,3\}$ is a~$3$-coloring of~$G$.
		\Wilog{} assume that~$c(v) = 1$.
		Then,~$c' \colon V'\rightarrow \{1,2,3\}$, with~$c'(u) \coloneqq c(u)$ for all~$u \in V \setminus \{v\}$,~$c'(u_i^1) = 1$ for all center vertices~$u_i^1$,~$c'(u_i^2) \coloneqq 2$ for all left vertices~$u_i^2$, and~$c'(u_i^3) \coloneqq3$ for all right vertices~$u_i^3$, is a valid~$3$-coloring of~$G'$.
		
		Conversely, suppose that~$c' \colon V' \rightarrow \{1,2,3\}$ is a~$3$-coloring of~$G'$.
		\Wilog, assume that~$c'(u_0^1)=1$. Then,~$\{c'(u_1^2), c'(u_1^3)\} = \{2,3\}$.
		Since~$u_1^1$ is adjacent to~$u_1^2$ and~$u_1^3$, it follows that~$c'(u_1^1) = 1$.
		By induction then~$c'(u_i^1)= 1$ for all~$i=1,\ldots,k$.
		Hence,~$c'(u) \neq 1~$ for all vertices~$u \in V \setminus \{v\}$ that are adjacent to~$v$ in~$G$, since they are adjacent to a center vertex in~$G'$.
		Then,~$c \colon V\rightarrow \{1,2,3\}$ with~$c(v)\coloneqq 1$ and~$c(u)\coloneqq c'(u)$ for all~$u \in V \setminus \{v\}$ is a valid~$3$-coloring of~$G$.
		\item Follows from (i).
		\item Suppose that~$c \colon V\rightarrow \{1,2,3\}$ is a~$3$-coloring of~$G$.
		Let~$c' \colon V'\rightarrow \{1,2,3\}$ with~$c'(v)\coloneqq c(v)$ for all~$v \in V$,~$c(w) \coloneqq 3$, and:
		\begin{align*}
			c'(u_i) \coloneqq 
			\begin{cases}
				1,& \text{ if } c(v_i) \neq 1,\\
				2,& \text{ if } c(v_i) = 1.
			\end{cases}
		\end{align*}
		Then,~$c'$ is a valid~$3$-coloring of~$G'$.
		
		If~$G'$ is~$3$-colorable, then so is~$G$, since~$G$ is a subgraph of~$G'$.
		
		\item The graph $G'$ contains $\bigO(n^2)$ corners each containing $\bigO(1)$ vertices and $\bigO(n^2)$ intermediate vertices.
		
		\end{enumerate}
\end{proof}

\begin{proof}[Proof of \cref{thm:tcNP}]
	Steps~2,~3, and~4 of \cref{constr:tc} preserve the $3$-colorability of $G'$ by \cref{lemma:tc}(i),~(ii), and~(iii), respectively.
	By \cref{lemma:tc}(iv), the size of the graph $G'$ output by the reduction is polynomial in the size of the input graph and it is easy to see that the computations in the construction may be performed in polynomial time.
	The degree restriction follows from the fact that the construction does not generate any vertices with degree above~7.
	
	It remains to show that $G'$ is an RNG and a GG.
	We start by giving an embedding of $G'$.
	If the $(x,y)$-corner is a single vertex, then its position is~$(x,y)$.
	If the~$(x,y)$-corner consists of a left and right vertex that are part of a
	coloring path added in step~2, then they are embedded
	at~$(x-\eps,y)$ and~$(x+\eps,y)$, respectively.
	If this corner consists of left and right vertices from a coloring path added in step~3,
	then their embedding is $(x,y+\eps)$ and $(x,y-\eps)$, respectively.
	Note that each intermediate vertex is adjacent to vertices in exactly two corners.
	If it is adjacent to the $(x, y)$- and $(x', y')$-corner,
	then it is embedded halfway in between, at $(\frac{x + x'}{2}, \frac{y + y'}{2})$.
		
	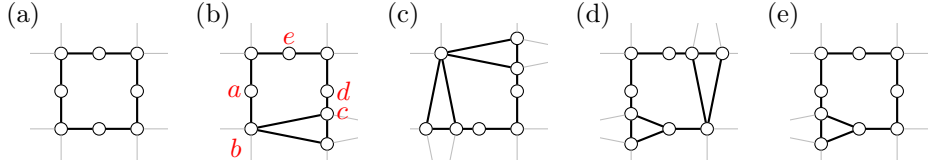
\begin{figure}[t]
		\centering
		\begin{tikzpicture}
		 \def\xr{1}
		 \def\yr{1}
		 \tikzpreamble{}
		 \def\teps{0.2}
		 \def\nsc{0.5}
		 
		 \def\xsh{2.5}
		 
		 \newcommand{\myclip}{\clip (0.6*\xr,0.6*\yr) rectangle (2.4*\xr,2.4*\yr);}
		 
		 \newcommand{\emptygrid}{
      \foreach\x in{0,1,2,3}{
        \foreach\y in{0,1,2,3}{
          \node (G\x\y) at (\x*\xr,\y*\yr)[inner sep=0pt]{};
        }
       }
		 }
		 
		 \begin{scope}
		  \node at (0.5*\xr,2.5*\yr)[]{(a)};
		  \myclip{}
		  \emptygrid{}
		  \foreach[count=\i, evaluate=\i as \z using int(\i-1)] \x/\y in {1/1,1/1.5,1/2,1.5/2,2/2,2/1.5,2/1,1.5/1}{
        \node (A\z) at (\x*\xr,\y*\yr)[xnode]{};
      }
      \foreach\x/\y in {G10/A0,G01/A0,A2/G02,A2/G13,A4/G23,A4/G32,A6/G31,A6/G20}{
        \draw[-,lightgray] (\x) to (\y);
      }
      \foreach\x/\y in {A0/A1,A1/A2,A2/A3,A3/A4,A4/A5,A5/A6,A6/A7,A7/A0}{
        \draw[xedge] (\x) to (\y);
      }
		 \end{scope}
		 
		 \begin{scope}[xshift=1*\xsh*\xr cm]
		  \node at (0.5*\xr,2.5*\yr)[]{(b)};
		  \myclip{}
		  \emptygrid{}
		  \foreach[count=\i, evaluate=\i as \z using int(\i-1)]\x/\y in {1/1,1/1.5,1/2,1.5/2,2/2,2/1.5,2/1+\teps,2/1-\teps}{
        \node (A\z) at (\x*\xr,\y*\yr)[xnode]{};
      }
      \foreach\x/\y in {G10/A0,G01/A0,A2/G02,A2/G13,A4/G23,A4/G32,A6/G31,A7/G31,A7/G20}{
        \draw[-,lightgray] (\x) to (\y);
      }
      \foreach\x/\y in {A0/A1,A1/A2,A2/A3,A3/A4,A4/A5,A5/A6,A6/A7,A7/A0,A6/A0}{
        \draw[xedge] (\x) to (\y);
      }
      \node at (A0)[anchor=north east,color=red]{$b$};
      \node at (A1)[anchor=east,color=red]{$a$};
      \node at (A3)[anchor=south,color=red]{$e$};
      \node at (A5)[anchor=west,color=red]{$d$};
      \node at (A6)[anchor=west,color=red]{$c$};
		 \end{scope}

		 \begin{scope}[xshift=2*\xsh*\xr cm]
		  \node at (0.5*\xr,2.5*\yr)[]{(c)};
		  \myclip{}
		  \emptygrid{}
		  \foreach[count=\i, evaluate=\i as \z using int(\i-1)] \x/\y in {1-\teps/1,1+\teps/1,1/2,2/2+\teps,2/2-\teps,2/1.5,2/1,1.5/1}{
        \node (A\z) at (\x*\xr,\y*\yr)[xnode]{};
      }
      \foreach\x/\y in {G10/A0,A1/G10,G01/A0,A2/G02,A2/G13,A3/G23,A4/G32,A3/G32,A6/G31,A6/G20}{
        \draw[-,lightgray] (\x) to (\y);
      }
      \foreach\x/\y in {A0/A1,A0/A2,A1/A2,A2/A3,A2/A4,A3/A4,A4/A5,A5/A6,A6/A7,A7/A1}{
        \draw[xedge] (\x) to (\y);
      }
		 \end{scope}
		 
		 \begin{scope}[xshift=3*\xsh*\xr cm]
		  \node at (0.5*\xr,2.5*\yr)[]{(d)};
		  \myclip{}
		  \emptygrid{}
		  \foreach[count=\i, evaluate=\i as \z using int(\i-1)] \x/\y in {1/1-\teps,1/1+\teps,1/2,1.5/2,2-\teps/2,2+\teps/2,2/1,1.5/1}{
        \node (A\z) at (\x*\xr,\y*\yr)[xnode]{};
      }
      \foreach\x/\y in {G10/A0,G01/A0,G01/A1,A2/G02,A2/G13,A4/G23,A5/G23,A5/G32,A6/G31,A6/G20}{
        \draw[-,lightgray] (\x) to (\y);
      }
      \foreach\x/\y in {A0/A1,A1/A2,A2/A3,A3/A4,A4/A5,A4/A6,A5/A6,A6/A7,A7/A1,A7/A0}{
        \draw[xedge] (\x) to (\y);
      }
  		\node () at (1*\xr,1.5*\yr)[xnode]{};
      \end{scope}
		 
		 \begin{scope}[xshift=4*\xsh*\xr cm]
		  \node at (0.5*\xr,2.5*\yr)[]{(e)};
		  \myclip{}
		  \emptygrid{}
		  \foreach[count=\i, evaluate=\i as \z using int(\i-1)] \x/\y in {1/1-\teps,1/1+\teps,1/2,1.5/2,2/2,2/1.5,2/1,1.5/1}{
        \node (A\z) at (\x*\xr,\y*\yr)[xnode]{};
      }
      \foreach\x/\y in {G10/A0,G01/A0,G01/A1,A2/G02,A2/G13,A4/G23,A4/G32,A6/G31,A6/G20}{
        \draw[-,lightgray] (\x) to (\y);
      }
      \foreach\x/\y in {A0/A1,A1/A2,A2/A3,A3/A4,A4/A5,A5/A6,A6/A7,A7/A1,A7/A0}{
        \draw[xedge] (\x) to (\y);
      }
  		\node () at (1*\xr,1.5*\yr) [xnode] {};
      \end{scope}

		\end{tikzpicture}

\caption{Grid faces created by \cref{constr:tc}.}
		\label{fig:tc-gf}
	\end{figure}
	
	We now show that the  RNG and GG induced by the vertices of any grid face
	is in fact the subgraph of~$G'$ induced by those vertices.
	To this end, we show for any pair of vertices sharing a grid face that there is no RNG blocker if they are adjacent and that there is a GG blocker if they are not adjacent.
	We do this by examining the grid faces individually.
	\cref{fig:tc-gf} pictures all types of grid faces that may occur in $G'$ (up to symmetry).
	
	In the case of the grid face pictured in \cref{fig:tc-gf}.(a) the claim is obvious.
	
	In the case of \cref{fig:tc-gf}.(b),
	the vertex $a$ is not an RNG blocker for the edge $\{b,c\}$, because
	\begin{align*}
		\dist(b,c)^2 & = 1 + \eps^2 \\
		\dist(a,c)^2 & = 1 + \left(\frac{1}{2} - \eps\right)^2 = 1 + \frac{1}{4} - 2\eps + \eps^2.
	\end{align*}
	Hence, if $\eps>0$ is sufficiently small then $\dist(b,c)^2 < \dist(a,c)^2$,
	implying that $a$ is not an RNG blocker for $\{b,c\}$. Similarly,
	\begin{align*}
		\dist(d,b)^2 = \dist(e,b)^2 = 1 + \frac{1}{4},
	\end{align*}
	implies that $d$ and $e$ also are not an RNG blockers for this edge if $\eps$ is sufficiently small.

	For the other grid faces, the claim follows along the same lines.
	
	By \cref{lemma:planar}, there can be no edges between vertices that do not share a grid face.
	Thus, we have proven that the given embedding induces~$G'$ as its RNG and GG.
	
	\emph{ETH-lower bound}. \tcAcr{} on planar graphs with maximum degree four
	cannot be solved in $2^{o(n^{1/2})}$~time where $n$ is the number of vertices unless \ETHbreaks{}.
	To show this, we inspect the reduction by \cite{Garey1976c}.
	First, it reduces \textsc{3-SAT} to~\tcAcr{} (on arbitrary graphs).
	This reduction maps formulas with $m$~clauses and $n$~variables to graphs with $\bigO(m+n)$~vertices and $\bigO(m+n)$~edges.
	\tcAcr{} on arbitrary graphs is then reduced to planar \tcAcr{},
	mapping graphs with $n$~vertices and $m$~edges to graphs with $\bigO(n+m^2)$~vertices and $\bigO(m^2)$~edges.
	Then, \tcAcr{} on arbitrary planar graphs is reduced to \tcAcr{} on planar graphs with maximum degree four.
	This reduction replaces each vertex~$v$ with a subgraph containing $\bigO(\deg(v))$~vertices and edges.
	Hence, this reduction maps graphs containing $n$~vertices and $m$~edges to graphs with $\bigO(m)$~vertices and edges.
The composition of these reductions yields a reduction from~\textsc{3-SAT} to~\tcAcr{} on planar graph with maximum degree four
	that maps a formula with $m$~clauses and $n$~variables to a graph containing $\bigO(m^2+n^2)$~vertices and edges.
\end{proof}

We remark that RNGs and GGs with maximum degree three are always $3$-colorable.
This follows from Brooks' theorem~\cite{Brooks41,Lovasz1975}, 
which states that any graph with maximum degree~$\Delta\geq 3$ 
which contains no $(\Delta +1)$-clique, 
is $\Delta$-colorable.
RNGs and GGs contain no $4$-cliques
(see~\cref{sec:prelim}).
It remains open whether \tcAcr{} can be solved in polynomial time on RNGs or GGs with maximum degree between four and six.

By the well-known $4$-color theorem,
all planar graphs are $4$-colorable.
The fastest known algorithm to compute a $4$-coloring of a planar graph
has quadratic running time~\cite{Robertson1996}.
For RNGs, 
Cimikowski~\cite{Cimikowski1990} proposed an algorithm for computing $4$-colorings in linear time.
However, 
this algorithm is based on the claim~\cite{Urquhart1983} that the wheel graph~$W_6$ cannot occur as subgraph of an RNG.
This claim was 
disproved by Bose \etal~\cite{Bose2012}.
Cimikowski's algorithm additionally implicitly assumes that RNGs are closed under minors, since the algorithm sometimes merges two adjacent vertices.
This can lead to graphs that are not RNGs.
Thus,
it remains open whether or not a linear-time algorithm for this task exists.

\section{Feedback Vertex Set}
\label{sec:fvs}

In this section, we will study the following problem:
\decprob{\fvsTsc{} (\fvsAcr{})}{fvs}
{An undirected graph~$G=(V,E)$ and~$k\in\Nzero$.}
{Is there a vertex set~$X\subseteq V$ with~$|X|\leq k$ such that~$G-X$ is a forest?}
\noindent
We will show:

\begin{theorem}
 \label{thm:fvsNP}
 \fvsTsc{} on RCGs,
 on RNGs, and on
 GGs
 is \NP-hard,
 even if the maximum degree is four.
 Unless~\ETHbreaks,
it
  admits no $2^{o(n^{1/4})}$-time algorithm
  where~$n$ is the number of vertices.
\end{theorem}

\noindent
Our proof is based on a polynomial-time many-one reduction
from the \NP-complete~\cite{Speckenmeyer1983} 
\fvsTsc{} on planar graphs of maximum degree four.

The second part of~\cref{thm:fvsNP} is derived from the following.

\begin{observation}
 \label{obs:fvsETH}
 Unless~\ETHbreaks,
  \fvsTsc{} on planar graphs of maximum degree four
  admits no $2^{o(n^{1/2})}$-time algorithm
  where $n$ is the number of vertices.
\end{observation}

\begin{proof}
 Speckenmeyer~\cite{Speckenmeyer1983} proves that \fvsAcr{} is \NP-hard on planar graphs with maximum degree four
	using a series of reductions starting from \textsc{Vertex Cover} on planar graphs with maximum degree four 
	and each of these reductions only introduces a linear change in the number of vertices.
	Using \cref{lemma:vc-ETH}, 
	this implies that \fvsAcr{} on planar graphs with maximum degree four admits no~$2^{o(n^{1/2})}$-time algorithm where $n$ is the number of vertices unless \ETHbreaks.
\end{proof}

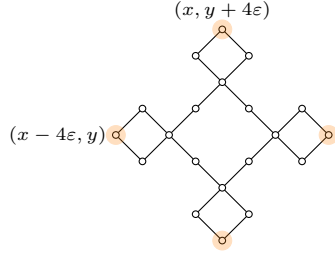
\begin{figure}
\centering
 \begin{tikzpicture}
  \def\xr{1}
  \def\yr{1}
  \tikzpreamble{}
  \def\teps{0.35}
  \def\bsc{0.8}
  \def\fsize{\scriptsize}
  
  \begin{scope}
    \buffer{A}{0}{0}{\bsc};
\node at (A!tCt)[anchor=south,font=\fsize]{$(x,y+4\eps)$};
    \node at (A!lCl)[anchor=east,font=\fsize]{$(x-4\eps,y)$};
    \foreach \z in {lCl,rCr,tCt,bCb}{\node at (A!\z)[xhiliS]{};}
\end{scope}
 
 \end{tikzpicture}
 \caption{Buffer at position~$(x,y)\in\R^2$ with highlighted outer vertices.}
 \label{fig:fvs-buffers}
\end{figure}

\noindent
In the reduction, we will use the graph pictured in \cref{fig:fvs-buffers}, which we call a \emph{buffer}.
The highlighted vertices are called \emph{outer vertices}.
We will also use several gadgets to represent vertices in the input graph.
The~$(4,0)$-, 
$(3,1)$-, 
and $(2,2)$-\emph{vertex gadgets} are pictured in \cref{fig:fvs-gadgets}.
Each of them contains several buffers.
The highlighted vertices in each vertex gadget are called \emph{outlets}.
When referring to them, we will order the outlets in the top and bottom half from left to right, calling them the first top outlet, second top outlet, etc.

\begin{figure}
 \centering
 \includegraphics{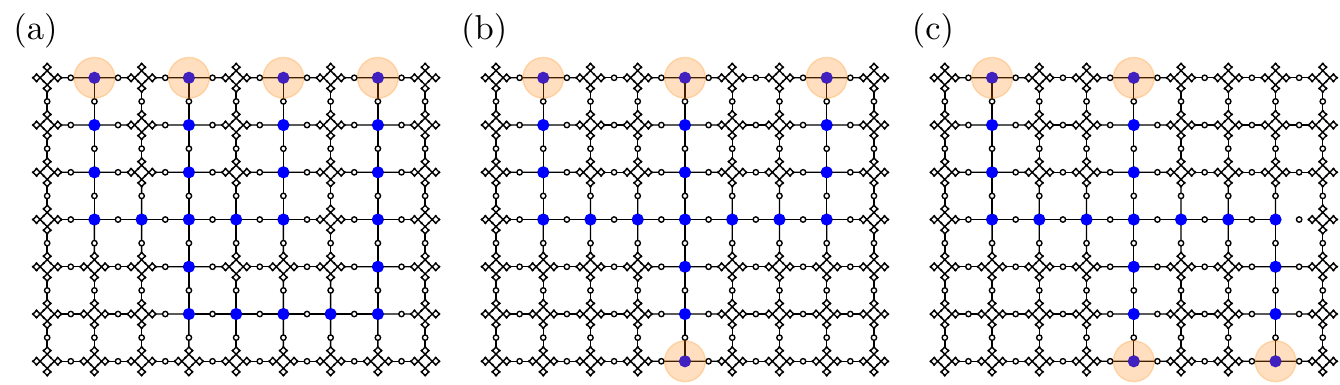}
 \caption{The (a) $(4,0)$-, 
 (b) $(3,1)$-, 
 and (c) $(2,2)$-vertex gadget with highlighted outlets.}
 \label{fig:fvs-gadgets}
\end{figure}

We will now give the construction of the reduction 
from \fvsAcr{} on planar graphs with maximum degree four
to \fvsAcr{} on RCGs, RNGs, and GGs.

\begin{figure}
	\centering
	\includegraphics{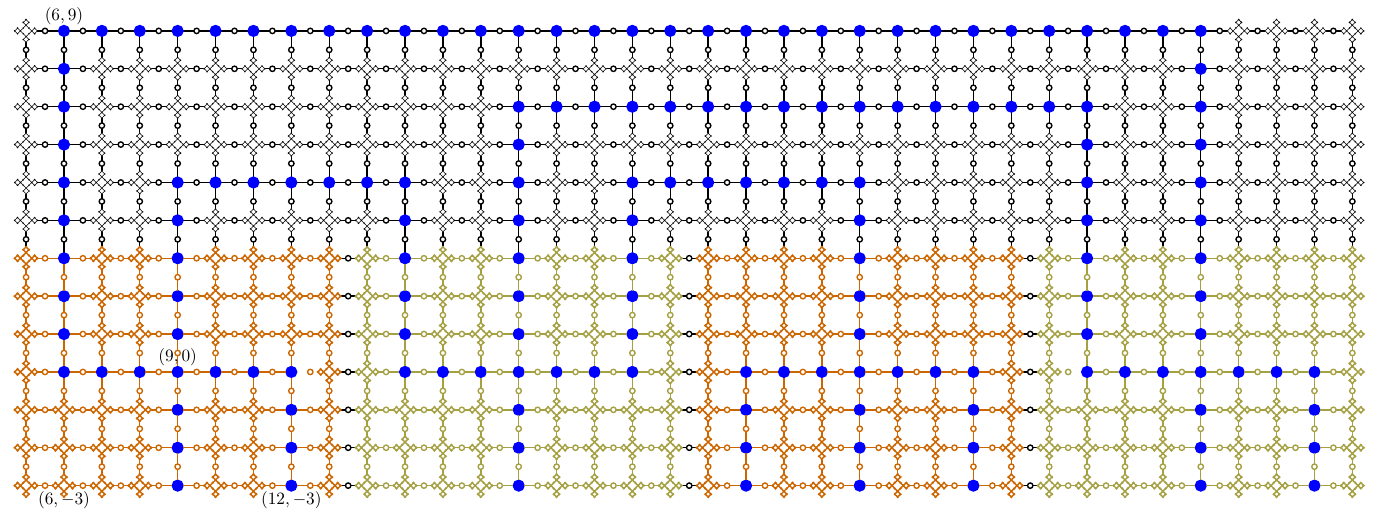}
	\caption{\cref{constr:fvs} applied to the graph in~\cref{fig:exgraph}.
		The large blue vertices form a subdivision of the original graph.
		The bottom half consists of four vertex gadgets,
		colored alternatingly orange and yellow. 
}
	\label{fig:fvs}
\end{figure}

\begin{construction}
	\label{constr:fvs}
	Let $G=(V,E)$~be a planar graph with maximum degree four and let~$k\in\Nzero$.
	We will construct a graph~$G'=(V',E')$ that is an RCG, 
	an RNG, 
	and a GG,
	and~$k'\in\Nzero$ such that $G$~contains a feedback vertex set of size~$k$ if and only if $G'$~contains a feedback vertex set of size~$k'$ (see \cref{fig:fvs} for an illustration).
	The vertex set~$V'$ will mostly consist of groups of vertices
	called $(x,y)$-\emph{corners},
	where $5 \leq x \leq 9n + 4$ and $-2h_2(G)-3 \leq y \leq 2h_1(G)+3$.
	Each corner consists of either a single vertex or a buffer.
	In the embedding of $G'$, the vertices forming the $(x,y)$-corner will be embedded roughly around the coordinate $(x,y)$.
	In the case of a buffer, only its outer vertices will have edges to vertices outside of that corner.
	We refer to them as the top, bottom, left, 
	and right \emph{connecting vertices} 
	of the corner.
	If a corner consists of a single vertex, then that vertex itself simultaneously forms the top, bottom, left, and right connecting vertex of that corner.
We start with $G'\coloneqq G$ and $k'\coloneqq k$ and then perform the following steps in order.

	\textbf{Step~1:} 
	Compute a 2-page book embedding of~$G$. 
	Let~$v_1,\ldots,v_n$
	be the vertices of~$G$ enumerated in the order in which they appear on the spine,
	and let~$\{E_1,E_2\}$ denote the partition of~$E$.

	\textbf{Step~2:} For each $v_i$ add an $(a,b)$-vertex gadget to $G'$ such that 
	$a+b=4$, 
	$a\geq\deg_1(v_i)$,
	and $b\geq\deg_2(v_i)$.
	Increase $k'$ by $4b$ where $b$ is the number of buffers in the vertex gadget.
	Connect the vertices on the boundary of each gadget to the adjacent gadgets as pictured in the lower half of \cref{fig:fvs}.
	The vertices of the gadget representing~$v_i$ are organized into $(x,y)$-corners
	where $9i - 4 \leq x \leq 9i + 4$ and $-3 \leq y \leq 3$,
	with the central vertex of that gadget forming the $(9i, 0)$-corner.
	In particular,
	the $a$~outlets and $9-a$~buffers along the top of the gadget
	form the $(x,3)$-corners where~$9i-4 \leq x \leq 9i+4$.
	Similarly, the $b$~outlets and $9-b$~buffers along the bottom of the gadget form the $(x, -3)$-corners.
	Recall the order of edges incident to a vertex defined in \cref{sec:prelim}.
	Let $e_1<\ldots<e_{\deg_1(v_i)}$ be the edges in~$E_1$ incident to~$v_i$ in that order.
	Then, 
	edge~$e_i$ is attached to the $i$-th top outlet.
	The edges in~$E_2$ are attached to the bottom outlets in the same manner.

	\textbf{Step~3:}
	For every edge~$e=\{v_i,v_j\}\in E_1$,~$i<j$, do the following:
	Suppose $e$~is $v_i$'s~$r$-th~edge and $v_j$'s~$s$-th~edge according to the order.
	After $e$ is attached to vertex gadgets in the previous step, $e$~starts at the top connecting vertex of the $(9i + \alpha, 3)$-corner and ends at the $(9j + \beta, 3)$-corner,
	where
	\begin{align*}
		\alpha \coloneqq
		\begin{cases}
			2r-5,& \text{ if } \deg_1(v_i)=4,\\
			3r-6,& \text{ if } \deg_1(v_i)\in\{2,3\},\\
0,& \text{ if } \deg_1(v_i)=1, 
		\end{cases}
		\qquad \text{and}
		\qquad 
		\beta \coloneqq
		\begin{cases}
			2s-5,& \text{ if } \deg_1(v_j)=4,\\
			3s-6,& \text{ if } \deg_1(v_j)\in\{2,3\},\\
0,& \text{ if } \deg_1(v_j)=1. 
		\end{cases}
	\end{align*}
Let~$L\ceq 9\ell(e) - \alpha + \beta - 1$ and $d\coloneqq 4h(e) + L$.
	Subdivide the edge $2d-1$ times and let $w_1,\ldots,w_{2d-1}$ be the vertices introduced by the subdivision.
	For every~$t \in \set{2h(e)}$, 
	the vertex~$w_{2t}$ is the $(9i+\alpha,3+t)$-corner,
	and if~$t<2h(e)$,
	the vertex~$w_{4h(e) + 2L + 2t}$ is the $(9j + \beta,3+2h(e) -t +1)$-corner.
	For every~$t\in\set{L}$, 
	the vertex~$w_{4h(e) + 2t}$ is the $(9i+\alpha+t,3+2h(e))$-corner.
If $t$ is odd, 
	then $w_t$~is an intermediate vertex.
	
	\textbf{Step~4:} Add buffers to~$G'$ as depicted in~\cref{fig:fvs}.
	More precisely, 
	for each~$(x,y)\in\N^2$ with~$5 \leq x \leq 9n + 4$ 
	and~$4 \leq y \leq 2h_1(G) + 3$,
	if no $(x,y)$-corner was added in step~3, 
	do the following.
	Add a buffer, 
	which becomes the $(x,y)$-corner,
	increase~$k'$ by~$4$,
	and add an edge connecting the top connecting vertex of the $(x,y)$-corner with the bottom connecting vertex of the $(x,y+1)$-corner,
    the left connecting vertex of the $(x,y)$-corner 
    with the right connecting vertex of the $(x-1,y)$-corner, 
    and so on.
    Subdivide each of these edges once,
    each resulting is an intermediate vertex.
    
	In steps~3~and~4, we have only taken $E_1$ into account.
	These steps must be repeated analogously for the edges in $E_2$,
	using negative $y$-coordinates.
	\cqed
\end{construction}

\noindent
In order to prove the correctness of the reduction, 
we will need several observations.

\begin{lemma}
	Let $G$ and $G'$ be graphs and let $\fvs(G)$ and $\fvs(G')$ denote the size of a smallest feedback vertex set in each graph.
	\begin{enumerate}[(i)]
		\item If $G'$ is obtained from~$G$ by subdividing an edge, then $\fvs(G) = \fvs(G')$.
		\item If $G'$ is obtained from~$G$ by adding vertices of degree at most one, then $\fvs(G) = \fvs(G')$.
		\item If $G'$ is obtained from~$G$ by adding a copy of the buffer graph and connecting each of the outer vertices of the buffer to at most one vertex in~$G$, then $\fvs(G') = \fvs(G) + 4$.
		\item If $G'$ is the graph obtained by applying \cref{constr:fvs} to~$G$,
		and $n$~is the number of vertices in~$G$, then $G'$ contains $\bigO(n^2)$~vertices.
	\end{enumerate}
	\label{lemma:fvs}
\end{lemma}
{
	\begin{proof}
		Parts~(i) and~(ii) are immediately obvious.
		\begin{enumerate}[(i)]
		\setcounter{enumi}{2}
			\item
  			Let $X$~be a feedback vertex set in $G$.
  			Then we obtain a feedback vertex set~$X'$ in~$G'$
  			by adding the four degree-4 vertices of the buffer to~$X$.
  			Thus $\fvs(G') \leq \fvs(G) +4$.
Conversely, any feedback vertex set~$X'$ in~$G'$ must contain at least 4 vertices from the buffer,
			as the buffer contains 4~vertex-disjoint cycles.
			As none of these 4~vertices intersects any cycle in~$G$,
			deleting them from~$X'$ yields a feedback vertex set for~$G$.
			Thus $\fvs(G) \leq \fvs(G') - 4$.

			\item The graph $G'$ contains~$\bigO(n)$ vertex gadgets
			(each containing a constant number of vertices)
			and~$\O(n^2)$ corners, each consisting of~$\O(1)$ vertices,
			and~$\O(n^2)$ vertices obtained from the subdivisions in step~5.
\qedhere
		\end{enumerate}
\end{proof}	
}

\noindent
We are set to prove the main result
of this section.

{
	\begin{proof}[Proof of~\cref{thm:fvsNP}]
Each vertex gadget consists of paths from a central vertex to the outlets as well as numerous buffers in addition to vertices that can be introduced by subdividing edges within the gadget.
		Hence, by \cref{lemma:fvs}, step~2 of \cref{constr:fvs} is correct.
		For steps 3~and~4, this also follows from the lemma.
		
		As noted in the same lemma, the size of~$G'$ is polynomial in the size of~$G$.
		It is easy to see that the computations in \cref{constr:fvs} can be performed in polynomial time
        and that there is no vertex of degree greater than~4.
	
		We will argue that $G'$ is an RNG.
		We embed the vertices in the following manner:
		The vertex gadget representing the vertex $v_i$ is embedded using the coordinates given in \cref{fig:fvs-gadgets}
		with the center vertex located at $(9i,0)$.
		The $(x,y)$-corner is embedded as shown in \cref{fig:fvs-buffers}
		with its center vertex
		located at~$(x, y)$.
		Intermediate vertices are adjacent to vertices in at most two corners and are embedded halfway between those two corners.
		
		\begin{figure}[t]
			\centering
            \begin{tikzpicture}
            \def\xr{1}
            \def\yr{1}
            \tikzpreamble{}
            
            \def\fname{figs/fvs-buffer.pdf}
            \def\fhname{figs/fvs-hbuffer.pdf}

            \def\nsc{0.4}
            \def\lsc{1}
            \def\bsc{0.25}
            \def\lcol{red}
            
            \newcommand{\myclip}{\clip (-0.4*\xr,-0.4*\yr) rectangle (1.4*\xr,1.4*\yr);}
            
            \newcommand{\fvsgrid}{
            \Grid{A}{3}{3}{-1}{-1}{scale=0.1}{-,lightgray}{1}{1}
            \Grid{B}{1}{1}{0}{0}{scale=0.1}{xedge}{1}{1}
            }
            
            \newcommand{\fvsplace}[1]{
                \foreach\x\y in {0/0,1/0,1/1,0/1}{
                    \node at (\x*\xr,\y*\yr)[xnode,scale=1.5,fill=blue]{};
                }    
                \foreach\x\y in {#1}{
                    \node at (\x*\xr,\y*\yr)[]{\includegraphics[scale=\xr]{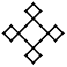}};
                }
                \foreach\x\y in {0/0.5,1/0.5,0.5/1,0.5/0}{
                    \node at (\x*\xr,\y*\yr)[xnode]{};
                }
            }
            
            \begin{scope}[shift={(0*\xr cm,0*\yr cm)}]
            \node at (-0.5*\xr,1.5*\yr)[]{(a)};
\myclip{}
            \fvsgrid{}
            \fvsplace{0/0,1/0,1/1,0/1}
            \end{scope}
            
            \begin{scope}[shift={(2.5*\xr cm,0*\yr cm)}]
            \node at (-0.5*\xr,1.5*\yr)[]{(b)};
\myclip{}
            \fvsgrid{}
            \fvsplace{0/0,1/1,0/1}
            \end{scope}
            
            \begin{scope}[shift={(5*\xr cm,0*\yr cm)}]
            \node at (-0.5*\xr,1.5*\yr)[]{(c)};
\myclip{}
            \fvsgrid{}
            \fvsplace{0/0,0/1}
            \end{scope}
            
            \begin{scope}[shift={(7.5*\xr cm,0*\yr cm)}]
            \node at (-0.5*\xr,1.5*\yr)[]{(d)};
\myclip{}
            \fvsgrid{}
            \fvsplace{0/0}
            \end{scope}
\end{tikzpicture}
\caption{The four types of grid faces in $G'$}
			\label{fig:fvs-gf}
		\end{figure}
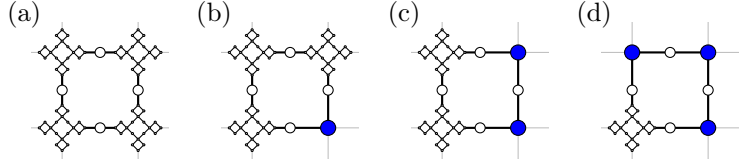
		
		We must show that there is no RCG blocker for any edge in $G'$ and that there is a GG blocker for every pair of non-adjacent vertices.
		For any $(x,y) \in \Z^2$ with $5\leq x < 9n+4$ and~$-2h_2(G) - 3\leq y < 2h_1(G) +3$, the subgraph of $G'$ induced by the vertices in the corners~$(x,y)$, $(x+1,y)$, $(x,y+1)$, and $(x+1,y+1)$ along with the intermediate vertices adjacent to vertices in those corners will be called a \emph{grid face}.
		We start by proving the claim for any two vertices that share a grid face.
		There are four types of grid faces (up to symmetry), depending on how many of its corners .
		They are pictured in \cref{fig:fvs-gf}.
		In each case, the intermediate vertices are GG blockers for any edges between corners.
\end{proof}

}

Note that \fvsAcr{} is polynomial-time solvable on graphs with maximum degree three~\cite{Ueno1988}.

\section{Hamiltonian Cycle}
\label{sec:hc}

In this section, we study the \hcTsc{} problem on RNGs and GGs (for RCGs the complexity will be left open).
This problem is defined by:
\decprob{\hcTsc{} (\hcAcr{})}{hc}
{An undirected graph~$G=(V,E)$.}
{Is there a cycle in~$G$ visiting every vertex in~$V$ exactly once?}
We will prove the following:

\begin{theorem}
	\label{thm:hcNP}
	\hcTsc{}
	on RNGs and on GGs
	is \NP-hard,
	even if the maximum degree is four.
	Moreover, unless \ETHbreaks, 
it
	admits no~$2^{o(n^{1/4})}$-time~algorithm
	where $n$ is the number of vertices.
\end{theorem}

\noindent
To prove \cref{thm:hcNP},
we give a polynomial-time many-one reduction from the restriction of \hcTsc{} to $3$-regular planar graphs,
for which we have the following.

\begin{proposition}[\cite{Garey1976b},\cite{Lokshtanov2011}]
	\label{lemma:hc-eth}
	\hcTsc{} on $3$-regular planar graphs is \NP-hard
	and
	admits no $2^{o(n^{1/2})}$-time algorithm
	unless \ETHbreaks{}.
\end{proposition}

\noindent
The reduction in the proof of \cref{thm:hcNP} consists of two Hamiltonicity-preserving modifications: 
gadget expansion (\cref{ssec:gadexp}) and face filling (\cref{ssec:facefill}).

\subsection{Gadget Expansion}
\label{ssec:gadexp}

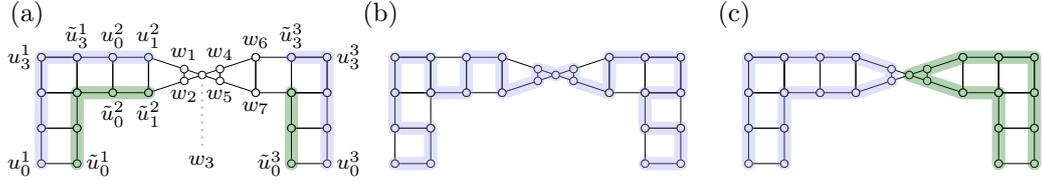
\begin{figure}[t]
\centering
	\begin{tikzpicture}
  \def\xr{0.47}
  \def\yr{0.47}
  \tikzpreamble{}
  \def\fpath{figs/hc-cycle.pdf}
  \def\ffpath{figs/hc-switch.pdf}
\def\onesix{1/6}

  \newcommand{\Ladder}[2]{\pgfmathsetmacro\they{int(#1 -2)}
    \pgfmathsetmacro\thex{int(#2 +5)}
    \pgfmathsetmacro\thexx{int(#2 +4)}
    
    \foreach\y in {0,...,\they}{
      \node at (0*\xr,\y*\yr)[]{\includegraphics[scale=\xr]{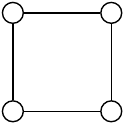}};
      \node at (\thex*\xr,\y*\yr)[]{\includegraphics[scale=\xr]{\fpath}};
    }
    
    \foreach\x in {0,...,#2,\thexx,\thex}{
      \node at (\x*\xr,\they*\yr)[]{\includegraphics[scale=\xr]{\fpath}};
    }
\pgfmathsetmacro\thexxx{#2 + 2}
    \node at (\thexxx*\xr,\they*\yr)[]{\includegraphics[scale=\xr]{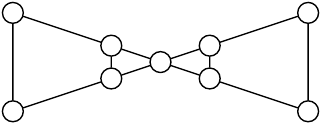}};
  }
  
  \newcommand{\LadderLab}[2]{\pgfmathsetmacro\they{int(#1 +1)}
    \pgfmathsetmacro\theyy{int(#1 -1)}
    \pgfmathsetmacro\thex{int(#2 +5)}
    \pgfmathsetmacro\thexx{int(#2 -1)}
    \begin{scope}[font=\footnotesize]
    \node at (-0.5*\xr,-0.5*\yr)[anchor=east]{$u_0^1$};
    \node at (0.5*\xr,-0.5*\yr)[anchor=west]{$\tilde{u}_0^1$};
    \node at (\thex*\xr-0.5*\xr,-0.5*\yr)[anchor=east]{$\tu_0^3$};
    \node at (\thex*\xr+0.5*\xr,-0.5*\yr)[anchor=west]{$u_0^3$};
\node at (-0.5*\xr,\theyy*\yr-0.5*\yr)[anchor=east]{$u_\theyy^1$};
    \node at (0.5*\xr,\theyy*\yr-0.5*\yr)[anchor=south]{$\tilde{u}_\theyy^1$};
    \node at (\thex*\xr-0.5*\xr,\theyy*\yr-0.5*\yr)[anchor=south]{$\tu_\theyy^3$};
    \node at (\thex*\xr+0.5*\xr,\theyy*\yr-0.5*\yr)[anchor=west]{$u_\theyy^3$};
    
    \node at (1.5*\xr,\theyy*\yr-0.5*\yr)[anchor=south]{$u_0^2$};
    \node at (1.5*\xr,\theyy*\yr-1.5*\yr)[anchor=north]{$\tilde{u}_0^2$};
    \node at (#2*\xr+0.5*\xr,\theyy*\yr-0.5*\yr)[anchor=south]{$u_\thexx^2$};
    \node at (#2*\xr+0.5*\xr,\theyy*\yr-1.5*\yr)[anchor=north]{$\tilde{u}_\thexx^2$};
    
    \node (c) at (#2*\xr+2*\xr,\theyy*\yr-3*\yr)[anchor=north]{$w_3$};
    \draw[-,dotted,lightgray,thick] (c) to (#2*\xr+2*\xr,\theyy*\yr-1.1*\yr);
    \node at (#2*\xr+1.5*\xr,\theyy*\yr-1*\yr+\onesix*\yr)[anchor=south]{$w_1$};
    \node at (#2*\xr+1.5*\xr,\theyy*\yr-1*\yr-\onesix*\yr)[anchor=north]{$w_2$};
    \node at (#2*\xr+2.5*\xr,\theyy*\yr-1*\yr+\onesix*\yr)[anchor=south]{$w_4$};
    \node at (#2*\xr+2.5*\xr,\theyy*\yr-1*\yr-\onesix*\yr)[anchor=north]{$w_5$};
    \node at (#2*\xr+3.5*\xr,\theyy*\yr-0.5*\yr)[anchor=south]{$w_6$};
    \node at (#2*\xr+3.5*\xr,\theyy*\yr-1.5*\yr)[anchor=north]{$w_7$};
    \end{scope}
  }
  
  \newcommand{\LadderDPaths}[2]{
    \pgfmathsetmacro\they{int(#1 -1)}
    \pgfmathsetmacro\thexx{int(#2 +5)}
    \pgfmathsetmacro\thex{int(#2 +1)}
    
    \draw[xpathS] (-0.5*\xr,-0.5*\yr) to (-0.5*\xr,\they*\yr-0.5*\yr) to (\thex*\xr-0.5*\xr,\they*\yr-0.5*\yr) to (\thex*\xr+0.5*\xr,\they*\yr-1*\yr+\onesix*\yr)
    to  (\thex*\xr+0.5*\xr,\they*\yr-1*\yr-\onesix*\yr) to (\thex*\xr-0.5*\xr,\they*\yr-1.5*\yr) to (0.5*\xr,\they*\yr-1.5*\yr) to (0.5*\xr,-0.5*\yr) ;
    \draw[xpathSx] (\thexx*\xr-0.5*\xr,-0.5*\yr) to (\thexx*\xr-0.5*\xr,\they*\yr-1.5*\yr) to (\thexx*\xr-1.5*\xr,\they*\yr-1.5*\yr)
    to (\thexx*\xr-3*\xr,\they*\yr-1*\yr) 
    to (\thexx*\xr-1.5*\xr,\they*\yr-0.5*\yr) to (\thexx*\xr-0.5*\xr,\they*\yr-0.5*\yr) to (\thexx*\xr+0.5*\xr,\they*\yr-0.5*\yr) to (\thexx*\xr+0.5*\xr,-0.5*\yr);
  }
  
  \newcommand{\mylab}[1]{\node at (-0.9*\xr,3.7*\yr)[]{#1};}
  
  \begin{scope}
    \mylab{(a)};
    \Ladder{4}{2}
    \LadderLab{4}{2}
    \draw[xpathS] (-0.5*\xr,-0.5*\yr) -- (-0.5*\xr,2.5*\yr) -- (2.5*\xr,2.5*\yr);
    \draw[xpathS] (6.5*\xr,2.5*\yr) -- (7.5*\xr,2.5*\yr) -- (7.5*\xr,-0.5*\yr);
    \draw[xpathSx] (0.5*\xr,-0.5*\yr) -- (0.5*\xr,1.5*\yr) -- (2.5*\xr,1.5*\yr);
    \draw[xpathSx] (6.5*\xr,1.5*\yr) -- (6.5*\xr,-0.5*\yr);
  \end{scope}
  
  \def\xsh{9.9}
  \begin{scope}[shift={(1*\xsh*\xr cm,0*\yr cm)}]
    \mylab{(b)};
    \Ladder{4}{2}
\draw[xpathS] (-0.5*\xr,-0.5*\yr) to (0.5*\xr,-0.5*\yr) to (0.5*\xr,0.5*\yr) to (-0.5*\xr,0.5*\yr) to (-0.5*\xr,2.5*\yr) to (0.5*\xr,2.5*\yr) to (0.5*\xr,1.5*\yr) to (1.5*\xr,1.5*\yr) to (1.5*\xr,2.5*\yr) to (2.5*\xr,2.5*\yr) to
    (2.5*\xr,1.5*\yr) to (3.5*\xr,2*\yr - \onesix*\yr) to (3.5*\xr,2*\yr + \onesix*\yr) to (4*\xr,2*\yr)  to (4.5*\xr,2*\yr+\onesix*\yr) to (4.5*\xr,2*\yr-\onesix*\yr) to (5.5*\xr,1.5*\yr) to (5.5*\xr,2.5*\yr) 
    to (6.5*\xr,2.5*\yr) to (7.5*\xr,2.5*\yr) to (7.5*\xr,1.5*\yr) to (6.5*\xr,1.5*\yr) 
    to (6.5*\xr,0.5*\yr) to (7.5*\xr,0.5*\yr) to (7.5*\xr,-0.5*\yr) to (6.5*\xr,-0.5*\yr); 
  \end{scope}
  
  \begin{scope}[shift={(2*\xsh*\xr cm,0*\yr cm)}]
    \mylab{(c)};
    \Ladder{4}{2}
    \LadderDPaths{4}{2}
  \end{scope}

  \end{tikzpicture}
  
	\caption{The ladder path~$L_{4,2}$ with 
	(a) selected vertex labels and inside/outside edges highlighted in dark green / light blue;
	(b) a traversal;
	(c) a partial/full~cover (light blue / dark green).}
	\label{fig:hc-ladderpath}
\end{figure}

The gadgets that will replace the edges are called \emph{ladder paths}.
For $k_1,k_2 \in \N$, the \emph{ladder path} with length $(k_1,k_2)$ is the graph $L_{k_1,k_2} = (V_{k_1,k_2}, E_{k_1,k_2})$,
where
\[
	V_{k_1,k_2}  = \{u^1_i, \tu^1_i, u^3_i, \tu^3_i \mid i \in \{0,\ldots,k_1-1\} \}
\cup \{u^2_i,\tu^2_i \mid i \in \{0,\ldots,k_2-1\}\}
	\cup \{w_1,\ldots,w_7\}
\]
and the edges are given using the example pictured in \cref{fig:hc-ladderpath}(a).
The vertices~$u^j_i, \tu^j_i$ with~$j\in\{1,2\}$ along with~$w_1$ and~$w_2$ form the \emph{first half} of the ladder path and those with~$j =3$ along with $w_4,\ldots,w_7$ form the \emph{second half}.
The vertex $w_3$ is the transitional vertex.
The vertices~$w_1,\ldots,w_7$ form the \emph{switch}.
The vertices~$u_0^1$ and~$\tu_0^1$ form the \emph{end} of the first half, 
while $u_0^3$ and $\tu_0^3$ form the \emph{end} of the second half. 
The edges highlighted in light blue in \cref{fig:hc-ladderpath}(a) will be called \emph{outside edges}, 
while the edges highlighted in dark green are \emph{inside edges}.
An edge~$\{u^i_j,u^i_{j+1}\}$ or~$\{\tilde{u}^i_j,\tilde{u}^i_{j+1}\}$ is called \emph{even} if~$j$ is even.

A \emph{traversal} of a ladder path is a path that begins in either vertex at one end of the ladder path, terminates in either vertex at the other end, and visits every vertex on the ladder path and no other vertex.
A \emph{partial cover} of a half of a ladder path is a path that begins in either vertex in the end of the half, terminates in the other vertex in that end, and visits every vertex of that half, but no other vertex.
A \emph{full cover} of a half additionally visits the transitional vertex. Examples of a traversal, a partial, and a full cover are pictured in \cref{fig:hc-ladderpath}(b)~and~(c).
The main property of ladder paths is that any Hamiltonian cycle must either contain a traversal, or a full and a half cover of each ladder path.

We will use the following minor technical lemma:

\begin{lemma}
	\label{lemma:hc-1}
	Suppose~$v_0,\ldots,v_{n-1}$ is a Hamiltonian cycle in~$G$. 
	Then, 
	for any~$i,j\in\set{n-2}$, 
	$i\leq j$, 
	the graph~$G-\{v_i,\ldots,v_j\}$
	contains no vertex of degree zero or one, 
	except possibly~$v_{i-1}$ and~$v_{j+1}$.
\end{lemma}

\begin{proof}
	Let~$i,j\in\set{n-2}$, 
	$i\leq j$.
	Any~$v_k$, 
	$k\in\set[0]{n-1}\setminus\{i-1,\ldots,j+1\}$, 
	has at least two neighbors in~$G-\{v_i,\ldots,v_j\}$.
\end{proof}

\begin{lemma}
	\label{lemma:hc-trav}
	Suppose that the Hamiltonian graph~$G=(V,E)$ contains a ladder, that the only vertices on the ladder path with neighbors outside of the ladder path are on its ends, and that the vertices on the ends each have no more than one neighbor outside of the ladder path.
	Then, any Hamiltonian cycle in~$G$ contains either:
	\begin{itemize}
		\item a traversal of the ladder path or
		\item a partial cover of one of its halves and a full cover of the other half.
	\end{itemize}
\end{lemma}

\begin{proof}
	Consider any Hamiltonian cycle~$C$ in~$G$ and any copy of the ladder path $L_{k_1,k_2}$ in $G$.
	In this Hamiltonian cycle, 
	the vertex $w_3$ is succeeded by $w_1$, $w_2$, $w_4$, or $w_5$.
	By symmetry, 
	we assume \wilog{} that the successor is~$w_1$ or $w_2$.
	We will only deal with~$w_1$ being the successor as the argument for~$w_2$ is very similar.
	The successor of~$w_1$ is in turn either~$w_2$ or~$u_{k_2-1}^2$.
	
	For the first case, 
	assume that $w_1$'s successor is~$w_2$.
	Then $w_2$~must be succeeded by~$\tu^2_{k_2-1}$,
	Since its only other neighbor is~$w_3$.
	By \cref{lemma:hc-1}, the successor of~$\tu^2_{k_2-1}$ must then be~$u^2_{k_2-1}$.
	By iterating this argument, 
	we can show that $C$~visits every vertex of the first half of the ladder path
	before leaving it at~$u^1_0$ or~$\tilde{u}^1_0$.
	
	If we look from~$w_3$ into the other direction,
	the transitional vertex~$w_3$ is preceded in~$C$ by 
	either~$w_4$ or~$w_5$.
By the same argument employed above, 
	we see that $C$~enters the second half of the ladder path at~$u^3_0$ or~$\tilde{u}^3_0$
	and then visits every vertex of the second half before reaching~$w_3$.
	
	Taking the two halves together,
	this shows that $C$ contains a traversal of the ladder path.
	
	Now, 
	for the second case, 
	assume that~$w_1$'s successor in~$C$ is~$u^2_{k_2-1}$.
	Then, 
	by \cref{lemma:hc-1}, 
	$w_3$'s predecessor must be~$w_2$ 
	and $w_2$'s predecessor in turn~$\tu^2_{k_2-1}$.
	From here it becomes clear that this argument can be continued to show that $C$~contains a full cover of the first half of the ladder path.
	By similar reasoning, $C$~must also contain a partial cover of the second half.
\end{proof}

\noindent
For a ladder path and a Hamiltonian cycle in a graph,
we say that
the ladder path is \emph{traversed} 
if the Hamiltonian cycle contains a traversal of the ladder path.
Otherwise, it is \emph{covered}.

Next, we discuss the vertex gadgets.
Recall that the graph~$G$ is assumed to be 3-regular.
We will use four types of vertex gadgets.
Each vertex gadget consist of a grid of size~$2 \times 10$,
with the only difference being the position of their three \emph{outlets},
which are designated vertex pairs to which the ladder paths representing the edges will be connected.
The \emph{$(3,0)$-vertex gadget} and the \emph{$(2,1)$-vertex gadget}
are pictured in \cref{fig:hc-gadgets} with the three outlets highlighted.
The \emph{$(0,3)$-vertex gadget} and the \emph{$(1,2)$-vertex gadget}
are obtained from the former two by mirroring along the horizontal axis.
The value $(i,j)$ will be called the \emph{type} of the gadget.

\begin{figure}
	\centering
	\includegraphics[width=\textwidth]{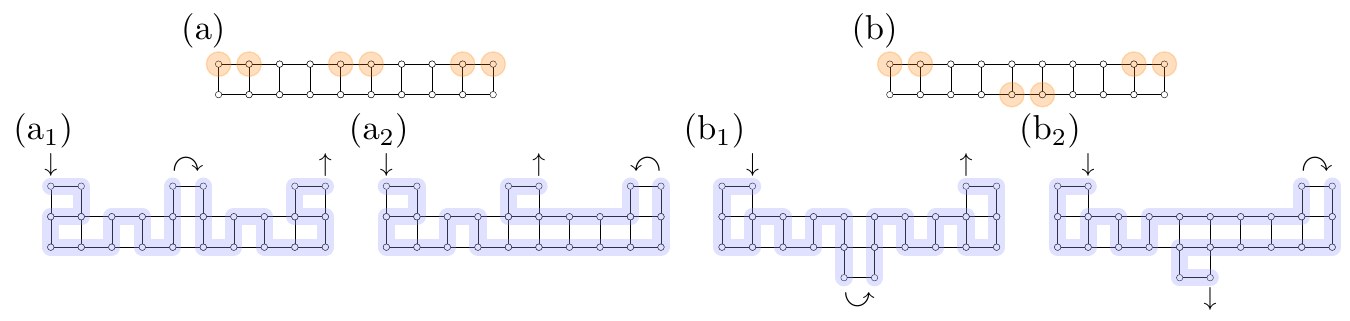}
	\caption{(a) $(3,0)$-vertex gadget and 
	(b) $(2,1)$-vertex gadget,
	with highlighted outlets.
	(a\textsubscript{1}), (a\textsubscript{2}) 
	and 
	(b\textsubscript{1}), (b\textsubscript{2}) 
	show two generic ways a Hamiltonian cycle can pass through each
	vertex gadget.
	}
	\label{fig:hc-gadgets}
\end{figure}

We will refer to the outlets as \emph{top} or \emph{bottom} outlets,
as well as the \emph{left}, \emph{middle} or \emph{right} outlet,
with the obvious meaning.
The left and right outlets are also called the \emph{outer outlets}.

We will now define the \emph{gadget expansion} of a 3-regular graph~$G$,
consisting of a graph~$G'$ and a straight-line embedding~$\emb$ resulting from applying the following steps to~$G$.

\begin{figure}
	\centering
	\includegraphics[width=\textwidth]{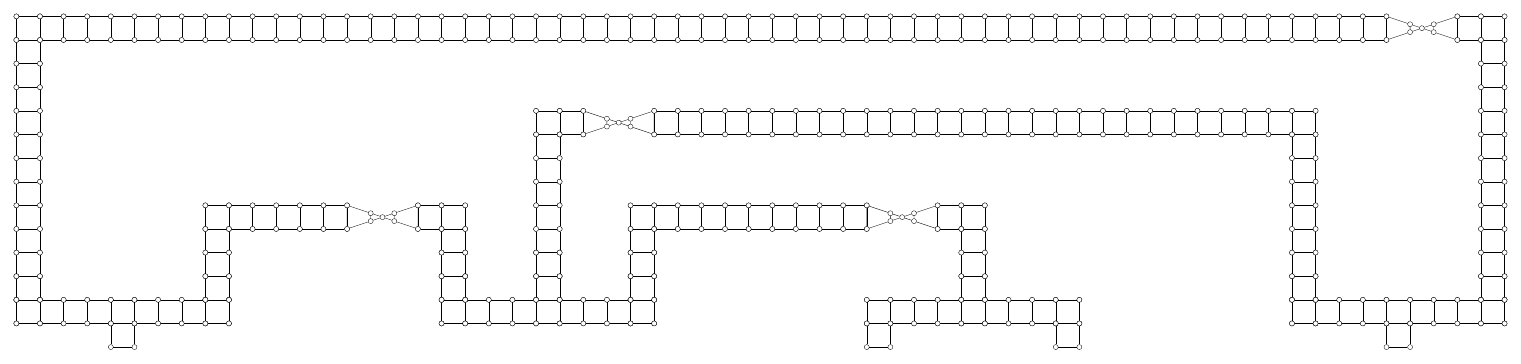}
	\caption{Gadget expansion of the graph pictured in \cref{fig:exgraph}. This graph is not~3-regular, but we may assume that there are further edges in $E_2$.}
	\label{fig:hc-prefinal}
\end{figure}

\begin{construction}[Gadget expansion]
	Start with~$G'$ being the empty graph (see \cref{fig:hc-prefinal} for an illustration).
	
	\textbf{Step 1:} Compute a 2-page book embedding of~$G$ and let~$v_1,\ldots,v_n$
	be the vertices of~$G$ in the order in which they appear on the spine,
	and let~$\{E_1,E_2\}$ denote the partition of~$E$.
	
	\textbf{Step 2:} For every vertex $v_i \in V$ add to~$G'$ a $(\deg_1(v_i),\deg_2(v_i))$-vertex gadget.
	Position the vertices of this gadget at $(18i+x,y)$ with $x \in\{0,\ldots,9\}$ and $y \in \{0,1\}$,
	as in \cref{fig:hc-gadgets}.
	
	\textbf{Step 3:} For every edge~$e=\{v_i,v_j\}$ in~$E_1$,~$i<j$,
	add to~$G'$ a ladder path~$L_{k_1,k_2}$ connected to an outlet in~$v_i$'s vertex gadget and an outlet in~$v_j$'s vertex gadget
	as follows.
	Recall the ordering of the edges incident to a vertex defined in the preliminaries.
	If~$e$~is the $r$-th~edge at~$v_i$ and the $s$-th~edge at~$v_j$,
	then attach said ladder path to the $r$-th top outlet from the left of $v_i$'s vertex gadget 
	and to the $s$-th top outlet from the left of $v_j$'s vertex gadget.
	If only one of these two outlets is an outer outlet, then attach the end of the first half to that outlet
	and the end of the second half to the other (middle) outlet.
	(If the outlets are both outer or both middle outlets, then it does not matter which end of the ladder path is connected to which outlet.)
	This is done by adding two disjoint edges which connect the two vertices forming an end of the ladder path
	to the two vertices forming the corresponding outlet as in \cref{fig:hc-prefinal}.
	
	The value of~$k_1$ is chosen as $k_1 \coloneqq 6h(e)$
	and the value of~$k_2$ as follows.
	Define~$\alpha$ to be~$0, 4$ or $8$,
	if the ladder path is attached to the left, middle or right outlet of $v_i$'s vertex gadget, respectively.
	Define $\beta$ in the same manner for $v_j$.
	Set $k_2\coloneqq 18(j-i) - \alpha + \beta -5$.
	Note that~$k_2 \geq 5$ always holds.
	Finally, we will give the embedding of the ladder path's vertices, using the designations introduced in the definition of a ladder path.
	We only state the case where the first half is attached to~$v_i$,
	for the other case the coordinates are to be mirrored at a suitable vertical axis.
	The positions of the vertices in the first half are (with~$S_i^\alpha\ceq 18i + \alpha$):
	\begin{align*}
		\emb(u_r^1) &\coloneqq (S_i^\alpha, r + 2),
&
		\emb(\tu_r^1) 
		&
		\coloneqq (S_i^\alpha +1, r + 2),
		\quad\,
		r=0,\ldots,k_1-1
		\\
\emb(u_r^2) &\coloneqq (S_i^\alpha+2+r,k_1+1),
&
		\emb(\tu_r^2) 
		&
		\coloneqq (S_i^\alpha+2+r,k_1),
		\quad
		r=0,\ldots,k_2-1
		\\
\emb(w_1) &\coloneqq (S_i^\alpha+k_2+2, k_1 + 2/3),
&
		\emb(w_2) 
		&
		\coloneqq (S_i^\alpha+k_2+2, k_1 + 1/3),
		\\
		\emb(w_3) &\coloneqq (S_i^\alpha+k_2+2.5, k_1 + 1/2),
&
		\emb(w_4) 
		&
		\coloneqq (S_i^\alpha+k_2+3, k_1 + 2/3),
		\\
		\emb(w_5) &\coloneqq (S_i^\alpha+k_2+3, k_1 + 1/3),
&
		\emb(w_6) 
		&
		\coloneqq (S_i^\alpha+k_2+4, k_1 + 1),
		\\
		\emb(w_7) &\coloneqq (S_i^\alpha+k_2+4, k_1).
	\end{align*}
	Compare \cref{fig:hc-ladderpath}.
	The positions of the vertices in the second half are analogous to the first.
	
	Step~3 must be repeated for~$E_2$ 
	using negative $y$-coordinates.
	\cqed
\end{construction}

\noindent
This construction is useful due to the following:

\begin{lemma}
	\label{lemma:hc-gadget-insertion}
	Gadget expansion preserves Hamiltonicity.
\end{lemma}

\begin{proof}
	First, assume that $G'$ contains a Hamiltonian cycle~$C'$.
	By \cref{lemma:hc-trav}, $C'$~contains either a traversal or a partial and a full cover of each ladder path.
	We obtain a Hamiltonian cycle in~$G$ by including each edge~$e$ if and only if $C'$~contains a traversal of the ladder path representing~$e$.
	
	Conversely, assume that $G$ contains a Hamiltonian cycle~$C$.
	We construct a Hamiltonian cycle~$C'$ in $G'$ by including
	a traversal of each ladder path representing an edge contained in~$C$
	as well as a full cover of one half of each ladder path representing an edge not contained in~$C$ and a partial cover of the other half.
	We additionally add several edges in the vertex gadgets to~$C'$.
	Which edges we pick depends on the type of the vertex gadget
	and which of the attached ladder paths are traversed and which are covered.
	There are four main cases
	(note that at each gadget there must be exactly two traversed ladder paths):
	\begin{itemize}
    \setlength{\itemindent}{1em}
		\item[(a\textsubscript{1})] a $(3,0)$-vertex gadget with both outer ladder paths traversed;
		\item[(a\textsubscript{2})] a $(3,0)$-vertex gadget with the middle and one outer ladder path traversed;
		\item[(b\textsubscript{1})] a $(2,1)$-vertex gadget with both outer ladder paths traversed;
		\item[(b\textsubscript{2})] a $(2,1)$-vertex gadget with the middle and one outer ladder path traversed.
	\end{itemize}
Other cases are symmetric to one of these four.
	In each case, 
	the edges included in~$C'$ are highlighted in~\cref{fig:hc-gadgets}.
\end{proof}

\subsection{Face Filling}
\label{ssec:facefill}

In order to turn the gadget expansion of a graph into an RNG and GG, we need to add \emph{buffers}.
The challenge is doing this in a way that preserves Hamiltonicity. We call an edge~$e$~of~$G$ \emph{permissible} if $G$~is not Hamiltonian or if $G$~contains a Hamiltonian cycle that passes through~$e$.
\begin{lemma}
	\label{lemma:hc-subdision}
	Subdividing a permissible edge preserves Hamiltonicity.
	Moreover, 
	both edges resulting from the subdivision are permissible in the resulting graph.
\end{lemma}

\begin{proof}
	Let~$e$ be a permissible edge of~$G$ and $G'$ be obtained from~$G$ by subdividing~$e$.
	If $G$ is not Hamiltonion, then clearly $G'$ is neither.
	Otherwise, if the edge~$e$ is contained in a Hamiltonian cycle~$C$,
	then~$C$ induces a Hamiltonian cycle~$C'$ in~$G'$ which contains the two edges
	created by the subdivision.
\end{proof}

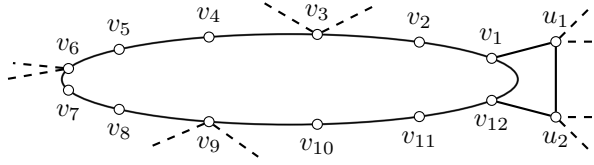
\begin{figure}[t]
	\centering
	\begin{tikzpicture}
	 \def\xr{1}
	 \def\yr{1}
	 \tikzpreamble{}
	 \def\nsc{0.4}
	 \newcommand{\theEllipse}[3]{
        \draw[fill=white,opacity=0.7,thick] (0,0) ellipse (#1*1 and #2*1);
        \draw (0,0) ellipse (#1*1 and #2*1);
        \foreach\x in {0,...,11}{
            \pgfmathsetmacro\xx{int(\x+1)}
            \ifnum\x<6\pgfmathsetmacro\angle{90}\else\pgfmathsetmacro\angle{-90}\fi
            \node (#3a\x) at ($(0,0)+(360/13*\xx:#1*1 and #2*1)$)[xnode,label=\angle:{$v_{\xx}$}]{};
        }
        \foreach\x in {0,...,30}{\node (#3c\x) at ($(0,0)+(360/30*\x:#1*1*1.25 and #2*1*1.75)$)[inner sep=0pt]{};}
        \foreach\x/\y in {2/6,2/8,5/15,5/14,8/20,8/22}{
            \draw[xedgex] (#3a\x) to (#3c\y);
        }
        }
     \theEllipse{3}{0.6}{A}
     \node (u1) at (3.5*\xr,0.5*\yr)[xnode,label=90:{$u_1$}]{};
     \node (u2) at (3.5*\xr,-0.5*\yr)[xnode,label=-90:{$u_2$}]{};
     \draw[xedge] (Aa0) to (u1) to (u2) to (Aa11);
     \node (u11) at ($(u1)+(0.5*\xr,0*\yr)$)[inner sep=0pt]{};
     \node (u12) at ($(u1)+(0.5*\xr,0.5*\yr)$)[inner sep=0pt]{}; 
     \draw[xedgex] (u11) to (u1) to (u12);
     \node (u21) at ($(u2)+(0.5*\xr,0*\yr)$)[inner sep=0pt]{};
     \node (u22) at ($(u2)+(0.5*\xr,-0.5*\yr)$)[inner sep=0pt]{}; 
     \draw[xedgex] (u21) to (u2) to (u22);
	\end{tikzpicture}

\caption{An example of a permissible cycle addition with a cycle of length~12.}
	\label{fig:hc-cycle}
\end{figure}

\noindent
Our main tool for adding buffers to a graph is called \emph{permissible cycle addition}.
Let $\{u_1, u_2\}$ be a permissible edge of~$G=(V,E)$.
We say that $G'=(V',E')$ is obtained from $G$ by \emph{attaching a permissible cycle to $\{u_1, u_2\}$} if (see~\cref{fig:hc-cycle} for an illustration)
\begin{itemize}
	\item $V' = V \uplus \{v_1,\ldots,v_k\}$, $k\geq 4$, and $v_1,\ldots,v_k$ induce a cycle in that order, that is, $\{v_i,v_j\} \in E'$ if and only if $|i-j| = 1$ or $\{i,j\} = \{1,k\}$;
	\item $E' \cap \binom{V}{2} = E$;
	\item for all $i \in \{2,\ldots,k-1\}$, if $\deg_{G'}(v_i) \geq 3$, then $\deg_{G'}(v_{i-1}) = \deg_{G'}(v_{i+1}) = 2$; and
	\item $\deg_{G'}(v_1) = \deg_{G'} (v_k) = 3$, and $\{\{v_1, u_1\}, \{v_k, u_2\}\} \subseteq E'$.
\end{itemize}
\noindent \cref{fig:hc-cycle} pictures an example of such a cycle addition.
This modification is useful due to:

\begin{lemma}
	\label{lemma:hc-cycle-addition}
	Permissible cycle addition preserves Hamiltonicity. 
	Moreover, 
	if $v_1,\ldots,v_k$ is the added cycle,
	then the edges~$\{v_i,v_{i+1}\}$, $1 \leq i < k$, are all permissible in the resulting graph.
\end{lemma}

\begin{proof}
	Suppose that $G'=(V',E')$ is obtained by permissible cycle addition from $G=(V,E)$ with $v_1,\ldots,v_k$ being the cycle added and $\{u_1,u_2\}$ the permissible edge in $G$ to which it is attached.
	
	First, suppose that $G'$ contains a Hamiltonian cycle.
	Due to~\cref{lemma:hc-1},
	this cycle cannot pass through any edge~$\{w,v_i\}$, $w\in V\setminus\{u_1,u_2\}$.
	Hence, 
	it must pass through $u_1,v_1,\ldots,v_k,u_2$ in that order or in reversed order.
	Thus, 
	all edges on the cycle except for possibly~$\{v_1,v_k\}$ are permissible.
	Removing~$v_1,\ldots,v_k$ from this cycle and adding the edge~$\{u_1,u_2\}$ 
	yields a Hamiltonian cycle in~$G$.
	
	Conversely, 
	suppose that $G$~contains a Hamiltonian cycle.
	Because edge~$\{u_1,u_2\}$ is permissible, 
	$G$~contains a Hamiltonian cycle that passes through~$\{u_1,u_2\}$.
	By inserting $v_1,\ldots,v_k$ between~$u_1$ and~$u_2$ we obtain a Hamiltonian cycle in~$G'$.
\end{proof}

\noindent
In order to be able to apply \emph{permissible cycle addition} to the gadget expansion~$G'$ of a graph~$G$,
we need to know permissible edges of~$G'$.
For this, we have the following lemma.

\begin{lemma}
	\label{lemma:hc-permissible-ladder}
	Let~$G$ be a 3-regular graph, $G'$~the gadget expansion of~$G$,
	and~$L$~any ladder path of~$G'$ whose first half is attached to an outer outlet (of a vertex gadget).
	Then,~$L$~contains two even inside and two even outside edges, all of which are permissible.
	Furthermore, these edges can be determined in linear time.
\end{lemma}

\begin{proof}
	We give the proof for the outside edges, for the inside edges everything works analogously.
	Write $L = L_{k_1,k_2}$.
	By construction of the gadget expansion, we have $k_1 \geq 6$ and $k_2 \geq 4$.
	Thus, using the vertex names from the definition of a ladder path,
	the edges $\{u^1_0, u^2_1\}, \{u^1_2, u^1_3\}, \{u^2_0, u^2_1\}$, and~$\{u^2_2, u^2_3\}$
	are even outside edges.
	
	It remains to show that at least two of these are permissible.
	So assume that $G'$ (and thus~$G$) is Hamiltonian (otherwise we are done)
	and let~$C$ be a Hamiltonian cycle of~$G$ and~$C'$~the corresponding Hamiltonian cycle of~$G'$ as given by \cref{lemma:hc-gadget-insertion}.
	By \cref{lemma:hc-trav}, $L$~is either covered or traversed by~$C'$.
	If it is covered, then~$C'$~contains every inside and outside edge of~$L$.
	If it is traversed, then the case analysis in \cref{fig:hc-gadgets} reveals that $C'$~contains either
	$\{u^1_0, u^2_1\}$ or $\{\tu^1_0, \tu^1_1\}$
	(which one depends on the type of the vertex gadget).
	In former case, $C'$~also contains $\{u^1_2, u^1_3\}$.
	In the latter case, $C'$ contains $\{u^2_0, u^2_1\}$ and $\{u^2_2, u^2_3\}$
	(cf.\ \cref{fig:hc-ladderpath}.(b)).
This proves the claim.
	Note that determining the two permissible edges does not require knowledge of~$C'$
	but only of the vertex gadget type.
\end{proof}

\noindent
We are set to give the construction in our polynomial-time many-one reduction from \hcAcr{} on $3$-regular planar graphs to \hcAcr{} on RNGs or GGs.

\begin{construction}
	\label{constr:hc}
	Let $G=(V,E)$ be a $3$-regular planar graph.
	We will construct an RNG and GG~$G'=(V',E')$ that is Hamiltonian if and only if $G'$~is.
	We will give the embedding of the vertices directly in the reduction.
	
	We start with ($G'$, $\emb$), the gadget expansion of~$G$.
We add one buffer for every~$(x, y)\in\Z^2$ where $14\leq x \leq 16n+10$ and $-6h_2(G) -2 \leq y \leq 6h_1(G)+2$
	are both even,
	except when~$G'$~already contains a vertex~$v$ 
	with~$\emb(v)\in \{(x, y),(x+1, y)\}$.
	The buffer then consists of a 4-cycle whose vertices are embedded at $(x, y), (x+1, y), (x+1, y+1)$, and $(x, y+1)$
	and whose edges are then further subdivided.
	We call it the $(x,y)$-buffer and refer to the four (subdivided) edges as its \emph{sides}.
	For each side, if $G'$ previously already contained a (possibly subdivided) edge running parallel to that side at distance~1,
	then we say that this side \emph{adjoins} that (possibly subdivided) edge.
	For example, the side from $(x, y+1)$ to $(x+1, y+1)$ would adjoin an existing edge from $(x, y+2)$ to $(x+1, y+2)$.
	A side may also adjoin an edge in a switch to which it is not parallel.
	For example, the side from $(x,y+1)$ to $(x+1,y+1)$ could adjoin an existing edge from $(x,y+2)$ to $(x+1,y+2+1/3)$.
	Finally, exactly one of the four sides will be designated as the \emph{docking side}.
	The docking side must adjoin either a side of a previously added buffer or a permissible edge of the gadget expansion.
	
	When adding a buffer, if its sides adjoin existing sides or edges,
	then we also add edges connecting the buffer to other vertices and possibly also subdivide the adjoined edges or sides.
	There are four cases, depending on whether the newly added side is docking or non-docking
	and whether it adjoins a side of another buffer or an edge of the gadget expansion.
	These four cases are illustrated in \cref{fig:hc-cycle-addition-constr}.
	In particular,
	\begin{itemize}
		\item the docking side of the added buffer is always subdivided four times and has four edges connecting it to the side or edge it adjoins (see \cref{fig:hc-cycle-addition-constr}(a) and (b));
		\item a non-docking side adjoining an edge of the gadget expansion is subdivided once and has two connecting edges (see \cref{fig:hc-cycle-addition-constr}(c));
		\item a non-docking side adjoining another buffer's side is subdivided thrice and has three connecting edges (see \cref{fig:hc-cycle-addition-constr}(d));
		\item a (non-docking) side which does not adjoin anything, then it is only subdivided once.
	\end{itemize}
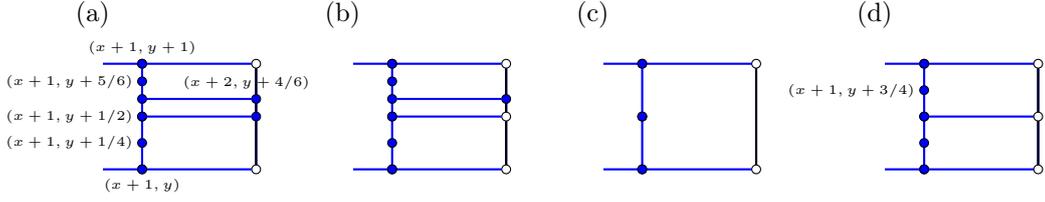
\begin{figure}[t]
		\centering
		\begin{tikzpicture}
		 \def\xr{1.4}
		 \def\yr{1.5}
		 \tikzpreamble{}
		 \def\nsc{0.35}
		 \newcommand{\myclip}[1]{\node at (-0.5*\xr,1.4*\yr)[]{(#1)};\clip (-0.4*\xr,-0.1*\yr) rectangle (1.2*\xr,1.4*\yr);}
        
        \newcommand{\dockgrid}{
            \Grid{A}{1}{2}{0}{-1}{}{-,xedge,color=blue}{1}{1}
\draw[xedge] (A02) to (A12);
            \node at (A01)[xnode,fill=blue]{};
            \node at (A11)[xnode,fill=blue]{};
            \node (c) at (0.5*\xr,0*\yr)[xnode,fill=blue]{};
        }
        
        \def\xsh{2.35}
        \def\fsize{\tiny}
        \begin{scope}
         \myclip{a}
         \begin{scope}[rotate around={-90:(0.5*\xr,0.5*\yr)}]
            \dockgrid{}
            \node (b1) at (1/6*\xr,0*\yr)[xnode,fill=blue]{};
            \node (b2) at (2/6*\xr,0*\yr)[xnode,fill=blue]{};
            \node (b3) at (3/6*\xr,0*\yr)[xnode,fill=blue]{};
            \node (b4) at (3/4*\xr,0*\yr)[xnode,fill=blue]{};
            \node (t1) at (2/6*\xr,1*\yr)[xnode,fill=blue]{};
            \node (t2) at (3/6*\xr,1*\yr)[xnode,fill=blue]{};
            \draw[xedge,color=blue] (b2) to (t1);
            \draw[xedge,color=blue] (b3) to (t2);
         \end{scope}
        \end{scope}
        
        \node at (A01)[anchor=south,font=\fsize]{$(x+1,y+1)$};
        \node at (A11)[anchor=north,font=\fsize]{$(x+1,y)$};
        \node at (b4)[anchor=east,font=\fsize]{$(x+1,y+1/4)$};
        \node at (b1)[anchor=east,font=\fsize]{$(x+1,y+5/6)$};
        \node at (b3)[anchor=east,font=\fsize]{$(x+1,y+1/2)$};
        \node at (t1)[anchor=south,font=\fsize]{$(x+2,y+4/6)$~\phantom{A}};

        \begin{scope}[shift={(\xsh*\xr,0*\yr)}]
         \myclip{b}
         \begin{scope}[rotate around={-90:(0.5*\xr,0.5*\yr)}]
         \dockgrid{}
         \node (b1) at (1/6*\xr,0*\yr)[xnode,fill=blue]{};
         \node (b2) at (2/6*\xr,0*\yr)[xnode,fill=blue]{};
         \node (b3) at (3/6*\xr,0*\yr)[xnode,fill=blue]{};
         \node (b4) at (3/4*\xr,0*\yr)[xnode,fill=blue]{};
         \node (t1) at (2/6*\xr,1*\yr)[xnode,fill=blue]{};
         \node (t2) at (3/6*\xr,1*\yr)[xnode]{};
         \draw[xedge,color=blue] (b2) to (t1);
         \draw[xedge,color=blue] (b3) to (t2);
        \end{scope}
        \end{scope}

        \begin{scope}[shift={(2*\xsh*\xr,0*\yr)}]
         \myclip{c}
         \begin{scope}[rotate around={-90:(0.5*\xr,0.5*\yr)}]
         \dockgrid{}
\node (b3) at (3/6*\xr,0*\yr)[xnode,fill=blue]{};
\end{scope}
        \end{scope}
        
        \begin{scope}[shift={(3*\xsh*\xr+0.3*\xr,0*\yr)}]
         \myclip{d}
         \begin{scope}[rotate around={-90:(0.5*\xr,0.5*\yr)}]
         \dockgrid{}
         \node (b1) at (1/4*\xr,0*\yr)[xnode,fill=blue]{};
\node (b3) at (3/6*\xr,0*\yr)[xnode,fill=blue]{};
         \node (b4) at (3/4*\xr,0*\yr)[xnode,fill=blue]{};
\node (t2) at (3/6*\xr,1*\yr)[xnode]{};
\draw[xedge,color=blue] (b3) to (t2);
        \end{scope}
        \end{scope}
        \node at (b1)[anchor=east,font=\fsize]{$(x+1,y+3/4)$};
\end{tikzpicture}

		\caption{Construction of the sides of a cycle (left) with
		(a) a docking side adjoined to an edge,
		(b) a docking side adjoined to a previously existing non-docking side,
		(c) a non-docking side adjoined to an edge, and
		(d) a non-docking side adjoined to an existing side.
In each picture,
an edge or previously existing side 
		to adjoin is 
		on the right. 
		The vertices and edges that result from the addition of the cycle are marked in blue while previously existing vertices and edges are in black.}
		\label{fig:hc-cycle-addition-constr}
	\end{figure}
\cref{fig:hc-cycle-addition-constr} explains the positions of the sides' subdivisions by way of an example for the right-side case.
 	For the case of the other three sides, 
 	the coordinates are obtained by rotating around $(x+1/2, y+1/2)$.
	
	The docking side must adjoin another buffer's side or a permissible edge.
	We will now discuss a strategy to achieve this.
	First, observe that every position~$(x, y)$ at which we intend to add a buffer lies in a face of $(G', \emb)$ that borders more than four vertices (possibly the unbounded face).
	We call these faces the \emph{regions}
	and will add the buffers region-by-region.
Next, note that once a buffer has been added to a region,
	then any subsequent buffer in that region can have its docking side adjoined to a (non-docking) side of another buffer added before it,
	as all edges on non-docking sides of a buffer are permissible by \cref{lemma:hc-cycle-addition}.
	(The gadget expansion is ``surrounded'' with buffers, so this works for the unbounded face.)
	Thus, it suffices to show how to add the first buffer for each region.
	
	To this end, 
	we must examine the structure of the gadget expansion.
	Let~$R$~be any region.
	Clearly, $R$~borders some vertex gadget,
	and thus also a ladder path~$L$ attached to an outside outlet of that vertex gadget.
More precisely, $R$~borders either every inside or every outside edge of~$L$.
	By construction of the gadget expansion, 
	the first half of~$L$ is attached to an outside outlet of some vertex gadget.
	Thus, we can find two permissible edges of~$L$ that border~$R$ by \cref{lemma:hc-permissible-ladder}.
	Because we have two permissible edges to choose from,
	we can ensure that we never adjoin the docking sides of two buffers to two ``parallel'' edges of~$L$
	(i.e., to $\{u^j_i, u^j_{i+1}\}$ and $\{\tu^j_i, \tu^j_{i+1}\}$),
	as every ladder path only borders two regions.
	\cqed
\end{construction}

\noindent
We are now prepared to prove the main result of this section.

\begin{proof}[Proof of \cref{thm:hcNP}]
	The proof builds on \cref{constr:hc}.
	By \cref{lemma:hc-gadget-insertion}, gadget expansion preserves Hamiltonicity.
	Each addition of a cycle involves subdividing a permissible edge
	(preserving Hamiltonicity by \cref{lemma:hc-subdision}),
	and then adding a permissible cycle (preserving Hamiltonicity by \cref{lemma:hc-cycle-addition}).
	It follows that the construction preserves Hamiltonicity.

	The construction already describes an embedding of the resulting graph~$G'$.
	So it only remains  to show that this embedding induces $G'$ as its RNG and GG.
	Let~$\mathcal{A}\ceq \{(x,y) \in \Z^2 \mid 14 \leq x \leq 16n + 11,\, -6h_2(G) -2 \leq y \leq 6 h_1(G)+3\}$.
	Note that for most $(x,y) \in \mathcal{A}$
there is a vertex embedded at $(x,y)$.
	The only exceptions are positions surrounding switches.
	For any $(x,y)\in\mathcal{A}$
we will call the vertices embedded at $(x,y)$, $(x+1,y)$, $(x,y+1)$, and $(x+1,y+1)$ along with any vertices embedded on the line segments between those four points a \emph{grid face}.
	There are three classes of grid faces: grid faces within ladder paths or vertex gadgets,
	buffers, 
	and grid faces between the aforementioned ones.
	
	Within ladder paths, only two grid faces, which are shown in \cref{fig:hc-emb} ($A$ and $B$), can occur.
	For these, it is easy to see that all nonadjacent vertex pairs have a GG blocker.
	Within buffers, more variations are possible (e.g. $C$ and $D$ in \cref{fig:hc-emb}).
	However, the vertices in the corners and in the center of each side serve as GG blockers for all nonadjacent pairs.
	All grid faces between cycles or between cycles and vertex gadgets are pictured in \cref{fig:hc-cycle-addition-constr}.
	Again, it is easy to see that all pairs of non-adjacent vertices have GG blockers
	and no pair of adjacent vertices has an RNG blocker.
	
	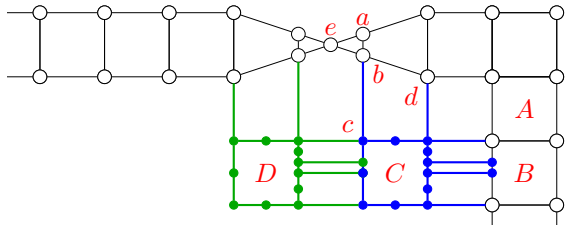
\begin{figure}[t]
		\centering
		\begin{tikzpicture}
            \def\xr{0.85}
            \def\yr{0.85}
            \tikzpreamble{}
            \def\fpath{figs/hc-cycle.pdf}
            \def\ffpath{figs/hc-switch.pdf}
\def\onesix{1/6}

            \newcommand{\Ladder}[2]{\pgfmathsetmacro\they{int(#1 -2)}
                \pgfmathsetmacro\thex{int(#2 +5)}
                \pgfmathsetmacro\thexx{int(#2 +4)}
                
                \foreach\y in {0,...,\they}{
                \node at (0*\xr,\y*\yr)[]{\includegraphics[scale=\xr]{\fpath}};
                \node at (\thex*\xr,\y*\yr)[]{\includegraphics[scale=\xr]{\fpath}};
                }
                
                \foreach\x in {0,...,#2,\thexx,\thex}{
                \node at (\x*\xr,\they*\yr)[]{\includegraphics[scale=\xr]{\fpath}};
                }
\pgfmathsetmacro\thexxx{#2 + 2}
                \node at (\thexxx*\xr,\they*\yr)[]{\includegraphics[scale=\xr]{\ffpath}};
            }
            \newcommand{\myclip}{\clip (3*\xr,1.2*\yr) rectangle (12*\xr,5*\yr);}
            
            \def\nsc{0.33}
            \newcommand{\grabber}[4]{
                \begin{scope}[shift={(#2*\xr,#3*\yr)}]
                    \node (#1x1) at (0*\xr,1*\yr)[xnode,color=#4]{};
                    \node (#1x2) at (1*\xr,1*\yr)[xnode,color=#4]{};
                    \node (#1x3) at (1*\xr,0*\yr)[xnode,color=#4]{};
                    \node (#1x4) at (0*\xr,0*\yr)[xnode,color=#4]{};
                    
                    \draw[xedge,color=#4] (#1x1) to (#1x2) to (#1x3) to (#1x4) to (#1x1);
                    
                    \node (#1x11) at (0.5*\xr,1*\yr)[xnode,color=#4]{};
                    \node (#1x21) at (1*\xr,1*\yr-1/6*\yr)[xnode,color=#4]{};
                    \node (#1x22) at (1*\xr,1*\yr-2/6*\yr)[xnode,color=#4]{};
                    \node (#1x22r) at (2*\xr,1*\yr-2/6*\yr)[xnode,color=#4]{};
                    \node (#1x23) at (1*\xr,1*\yr-3/6*\yr)[xnode,color=#4]{};
                    \node (#1x23r) at (2*\xr,1*\yr-3/6*\yr)[xnode,color=#4]{};
                    \node (#1x24) at (1*\xr,1*\yr-3/4*\yr)[xnode,color=#4]{};
                    \node (#1x31) at (0.5*\xr,0*\yr)[xnode,color=#4]{};
                    \node (#1x41) at (0*\xr,0.5*\yr)[xnode,color=#4]{};
                    
                    \draw[xedge,color=#4] (#1x22) to (#1x22r);
                    \draw[xedge,color=#4] (#1x23) to (#1x23r);
                \end{scope}
            }
            
            \def\argh{0.1}
            \begin{scope}
             \myclip{}
             \Ladder{6}{6}
\grabber{A}{8.5}{1.5}{blue};
             \grabber{B}{6.5}{1.5}{green!66!black};
             \draw[green!66!black,xedge] (Bx1) to ($(Bx1)+(0,1*\yr-\argh*\yr)$);
             \draw[green!66!black,xedge] (Bx2) to ($(Bx2)+(0,1*\yr+0.5*\yr-1/6*\yr-\argh*\yr)$);
             \draw[green!66!black,xedge] (Bx2) to (Ax1);
             \draw[green!66!black,xedge] (Bx3) to (Ax4);
             \draw[blue,xedge] (Ax1) to ($(Ax1)+(0,1*\yr+0.5*\yr-1/6*\yr-\argh*\yr)$);
             \draw[blue,xedge] (Ax2) to ($(Ax2)+(0,1*\yr-\argh*\yr)$);
             \draw[blue,xedge] (Ax2) to ($(Ax2)+(1*\xr-\argh*\xr,0)$);
             \draw[blue,xedge] (Ax3) to ($(Ax3)+(1*\xr-\argh*\xr,0)$);
             \node at (8.5*\xr,2*\yr)[xnode,color=blue]{};
             
             \node at (8.5*\xr,4*\yr+1/6*\yr)[anchor=south,color=red]{$a$};
             \node at (8.5*\xr,4*\yr-1/6*\yr)[anchor=north west,color=red]{$b$};
             \node at (8.5*\xr,2.5*\yr)[anchor=south east,color=red]{$c$};
             \node at (9.5*\xr,3.5*\yr)[anchor=north east,color=red]{$d$};
             \node at (8*\xr,4*\yr)[anchor=south,color=red]{$e$};
             
             \node at (11*\xr,3*\yr)[color=red]{$A$};
             \node at (11*\xr,2*\yr)[color=red]{$B$};
             \node at (9*\xr,2*\yr)[color=red]{$C$};
             \node at (7*\xr,2*\yr)[color=red]{$D$};
            \end{scope}

		\end{tikzpicture}
		\caption{An excerpt of the graph~$G'$ produced by the reduction:
		Grid faces in $G'$:
        (A) Grid face within a ladder path or a vertex gadget with no docking side adjoined to its edges.
        (B) Grid face within a ladder path or a vertex gadget with a docking side adjoined to one of its edges.
        (C) + (D) Grid faces within buffers with the docking side on the right.
		}
		\label{fig:hc-emb}
	\end{figure}
	We now consider the area surrounding a switch.
	This area is pictured in \cref{fig:hc-emb}.
The vertex $e$ is not a GG blocker for $\{a,b\}$, because
	$\dist(a,b)^2  = 1/9$ and
		$\dist(a,e)^2  = \dist(b,e)^2 = 5/18$.
The vertex marked $d$ is also not a GG blocker for $\{b,c\}$, 
	since
	$\dist(b,c)^2=16/9$,
    $\dist(b,d)^2=10/9$,
    and
	$\dist(c,d)^2=2$.
Other cases are analogous or easy to see.
	Vertices that do not share a grid face are not adjacent by \cref{lemma:planar}.
	
	If $n$~is the number of vertices in the input graph~$G$, 
	then the graph~$G'$ output by the construction contains $n$~vertex gadgets each containing $\bigO(1)$~vertices, 
	$\bigO(n)$~ladder paths with $\bigO(n)$~vertices, 
	and $\bigO(n^2)$~cycles with $\bigO(1)$~vertices.
	It is easy to see that each step in both constructions can be computed in polynomial time. 
	Along with \cref{lemma:hc-eth}, 
	this implies that \hcAcr{} is \NP-hard on RNGs and on GGs.
	Moreover,
	it also implies that \hcAcr{} cannot be decided by a $2^{o(n^{1/4})}$-time algorithm on RNGs or GGs,
	unless~\ETHbreaks.
\end{proof}

\noindent
The computational complexity of \hcTsc{} on RCGs and that of \hcAcr{} on RNGs and GGs with maximum degree three is left open.

\section{Independent Set}
\label{sec:is}
In this section, we will investigate the restriction of the following problem to RNGs, RCGs, and Gabriel graphs:

\decprob{\isTsc{} (\isAcr{})}{is}
{An undirected graph~$G=(V,E)$ and~$k\in\Nzero$.}
{Is there a vertex set~$X\subseteq V$ with~$|X|\geq k$ such that~$G[X]$ is edgeless?}

We will show:

\begin{theorem}
 \label{thm:isNP}
 \isTsc{} on RCGs,
 on RNGs,
 and on GGs
 is \NP-hard,
 even if the maximum degree is four.
 Moreover, unless~\ETHbreaks,
 it admits
 no~$2^{o(n^{1/4})}$-time algorithm
 where $n$ is the number of vertices.
\end{theorem}

\noindent
Our proof is based on a polynomial-time many-one reduction 
from the \NP-hard~\cite{Garey1976} 
\isTsc{} on planar graphs with maximum degree three.

\begin{figure}[t]
    \centering
    \includegraphics{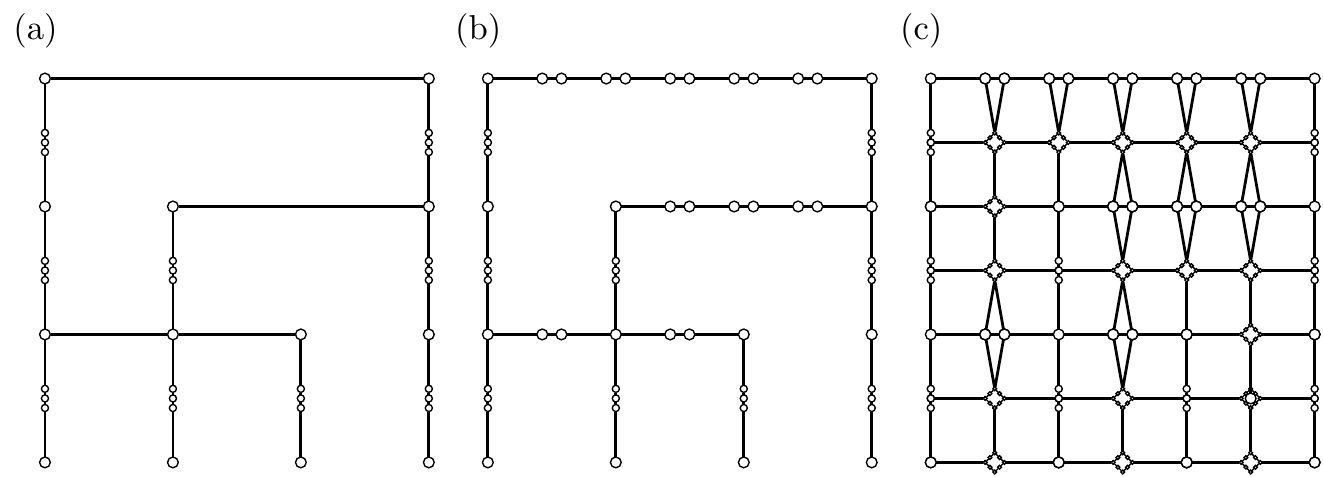}
    \caption{\cref{constr:is} 
applied to the example graph in~\cref{fig:exgraph}.
    (a) The graph after steps~1--2.
    (b) The graph after steps~1--3.
    (c) The final graph.
    }
    \label{fig:is}
\end{figure}

\begin{construction}
	\label{constr:is}
	Let $G=(V,E)$ be a planar graph with maximum degree three and let~$k\in \Nzero$.
	We construct a graph~$G'=(V',E')$' and $k'\in \Nzero$ such that $G'$~is an RNG and a GG (we will discuss RCGs later),
	and $G$ contains an independent set of size $k$ if and only if $G'$ contains one of size $k'$ (see \cref{fig:is} for an illustration).

	We start with $G'\coloneqq G$ and $k'\coloneqq k$.
	
	\textbf{Step 1:} Compute a 2-page book embedding of~$G$ and let~$v_1,\ldots,v_n$
	be the vertices of~$G$ enumerated in the order in which they appear on the spine,
	and let~$\{E_1,E_2\}$ denote the partition of~$E$.
	
	\textbf{Step 2:} Replace every vertex $v_i \in V$ with a path consisting of $4h_1(v_i) + 4h_2(v_i)+1$ vertices $w_{-4h_2(v_i)},\ldots,w_{4h_1(v_i)}$ and increase $k'$ by $2h_1(v_i) + 2h_2(v_i)$.
	For $j=-h_2(v_i),\ldots,h_1(v_i)$, the vertex $w_{4j}$ forms the $(2i,2j)$-corner.
	For $j=-h_2(v_i),\ldots,h_1(v_i)-1$, the vertices $w_{4j+1}, w_{4j+2},w_{4j+3}$ jointly form the $(2i,2j+1)$-corner,
	with $w_{4j+2}$ being both the left and right connecting vertex.
	
	Every edge~$e$ of~$G$ incident to~$v_i$ is now instead attached to $w_{4h(e)}$ (if~$e \in E_1$) or $w_{-4h(e)}$ (if~$e \in E_2$).
	
	\textbf{Step 3:}
	For every edge~$e = \{v_i, v_j\} \in E_1$, $i<j$,
	subdivide the corresponding edge of~$G'$ a total of $4\ell(e) - 2$~times and increase $k'$ by $2\ell(e) -1$.
	Suppose that $w_1,\ldots,w_{4\ell(e)-2}$ are the vertices introduced in the subdivision.
	For every~$x\in\set{2\ell-1}$, 
	the vertices~$w_{2x-1}$ and $w_{2x}$ jointly form the $(2i+x,2h(e))$-corner,
	with each being both 
	top and bottom connecting vertex.
	
	Repeat this step for~$E_2$ using negative $y$-coordinates instead.
	
	\textbf{Step 4:} For every $(x,y) \in \Z^2$ with $2\leq x \leq 2n$ and $-4h_2(G) \leq y \leq 4h_1(G)$,
	if no $(x,y)$-corner was added in any of the previous steps,
	then add a copy of~$\cei$,
	call it the $(x,y)$-corner,
	and increase~$k'$ by~$4$.
	The first, 
	third, 
	fifth, 
	and seventh vertex on the cycle are the top, 
	left, 
	bottom, 
	and right connecting vertex of that corner, 
	respectively.
	Make the top connecting vertex of the $(x,y)$-corner adjacent with the bottom connecting vertices of the $(x,y+1)$-corner,
	the left connecting vertex of the $(x,y)$-corner adjacent with the right connecting vertices of the $(x-1,y)$-corner, 
	and so on.
	\cqed	
\end{construction}

\begin{lemma}
	\label{lemma:is}
	Let $G=(V,E)$ and $G'=(V',E')$ be graphs and let $\is(G)$ and $\is(G')$ denote the size of a largest independent set in each graph.
	\begin{enumerate}[(i)]
		\item If $G'$ is obtained from $G$ by replacing a vertex~$v\in V$ with
		a path consisting of $2k+1$~new vertices $v_1,\ldots,v_{2k+1}$ in that
		order and connecting each $u \in N_G(v)$ to an arbitrary vertex in~$\{v_{2i+1} \mid 0\leq i \leq k \}$,
		then $\is(G') = \is(G) +k$.
		\item If $G'$ is obtained from $G$ by subdividing an edge $e \in E$ a
		total of $2k$ times, then~$\is(G') = \is(G)+k$.
		\item If $G'$ is obtained from $G$ by adding a copy of the cycle $\cei$
		such that no two consecutive vertices in the cycle have a neighbor outside of the cycle, then~$\is(G') = \is(G) +4$.
		\item If $G'$ is obtained from $G$ using \cref{constr:is} and $n$ is
		the number of vertices in $G$, then $G'$ contains at most $\bigO(n^2)$ vertices.
	\end{enumerate}
\end{lemma}
\begin{proof}
	\begin{enumerate}[(i)]
		\item
		\newcommand{\Vodd}{V_{\textnormal{odd}}}
		\newcommand{\Veven}{V_{\textnormal{even}}}
		Define $\Vodd := \{v_{2i+1} \mid 0\leq i \leq k\}$
		and $\Veven := \{v_{2i} \mid 1 \leq i \leq k\}$.
		
		Let $X\subseteq V$ be an independent set in~$G$.
		If $v \in X$, 
		then $(X \setminus \{v\}) \cup \Vodd$ is an independent set of size $|X|+k$ in $G'$.
		If $v\notin X$, 
		then $X \cup \Veven$ is an independent set of that size.
		Conversely, 
		let $X' \subseteq V'$ be an independent set in~$G'$.
		Note that $\abs{X' \cap \{v_1, \dots, v_{2k+1}\}} \leq k+1$
		with equality if and only if $\Vodd \subseteq X'$.
		If $\Vodd \subseteq X'$, 
		then $(X' \setminus \Vodd) \cup \{v\}$ is an independent set of~$G$ of size $\abs{X'}-k$.
		Otherwise, 
		$X' \setminus \{v_1, \dots, v_{2k+1}\}$ is an independent set of~$G$ of size at least~$\abs{X'}-k$.
		\item Follows from (i).
		\item Let $X \subseteq V$ be an independent set in $G$.
		The added copy of $\cei$ contains at least four pairwise non-adjacent vertices without any neighbors in $G$.
		Adding these four vertices to $X$ yields an independent set in $G'$ of size $|X|+4$.
		Conversely, let $X' \subseteq V'$ be an independent set in~$G'$.
		Since $X'$~can contain at most four vertices from the copy of~$\cei$,
		removing any such vertices yields an independent set in~$G$ of size at least~$|X'|-4$.
		\item The graph~$G'$ contains $\bigO(n^2)$ corners each containing~$\bigO(1)$ vertices.
	\end{enumerate}
\end{proof}
We are now ready to prove the main result of this section:
\begin{proof}[Proof of \cref{thm:isNP}]
	The correctness of steps~2~to~4 follows from \cref{lemma:is}(i) to~(iii), respectively.
	By \cref{lemma:is}(iv), the size of the output graph of the reduction is polynomial in the size of the input graph and it is easy to see that the computations in the construction may be carried out in polynomial time.
	The degree restriction follows from the fact that the reduction does not generate any vertices with degree greater than $4$.
	
	We will now argue that the graph $G'$ output by \cref{constr:is} is an RNG and a GG.
	We begin by giving an embedding of the graph.
	If the $(x,y)$-corner is a single vertex, then its embedding is $(x,y)$.
	If the $(x,y)$-corner consists of three vertices added in step~2, then they are embedded at $(x,y-\eps)$, $(x,y)$, and $(x,y+\eps)$.
	If the $(x,y)$-corner consists of two vertices added in step~3, then their positions are $(x\pm \eps,y)$.
	Finally, if the $(x,y)$-corner is a copy of $\cei$, then the vertices in this cycle are embedded at~$(x\pm \eps,y),(x\pm \frac{\eps}{2},y\pm\frac{\eps}{2}),(x,y\pm \eps)$.
	
	We must show that this embedding induces $G'$ as its RNG and as its GG.
	
	For any $(x,y)\in\Z^2$ with $2\leq x < 2n$ and $-4h_2(G) \leq y < 4h_1(G)$, 
	the $(x,y)$-\emph{grid face} is the set of all vertices
	contained in the corners~$(x,y)$, $(x+1,y)$, $(x,y+1)$, and $(x+1,y+1)$.
	By \cref{lemma:planar}, 
	vertices that do not share a grid face are non-adjacent in the GG induced by the embedding, 
	implying that they are also non-adjacent in the RNG.
	It remains to show that there is a GG blocker for each pair of non-adjacent vertices that do share a grid face, 
	but no RNG blocker for any edge.
	\cref{fig:is-gf} pictures all types of grid faces that may occur in~$G'$ (up to symmetry).
	
	\begin{figure}[t]
		\centering
		\begin{tikzpicture}
         \def\xr{0.94}
         \def\yr{0.94}
         \tikzpreamble{}
         \def\theeps{0.15}
         \def\tscale{0.66}
         \def\crss{\theeps}
         \def\crsc{0.35}
         \def\yeps{\theeps}
         \def\xeps{\theeps}
         
		 \newcommand{\cross}[8]{
            \begin{scope}[rotate around={90:(#1*\xr,#2*\yr)}]
                \fill[white] (#1*\xr+-1*\xr*\crss,#2*\yr-1*\yr*\crss) rectangle (#1*\xr+1*\xr*\crss,#2*\yr+1*\yr*\crss);
            \end{scope}

            \node (cL) at (#1*\xr+-1*\xr*\crss,#2*\yr+0*\yr*\crss)[xnode,scale=\crsc]{};
            \node (cT) at (#1*\xr+0*\xr*\crss,#2*\yr+1*\yr*\crss)[xnode,scale=\crsc]{};
            \node (cR) at (#1*\xr+1*\xr*\crss,#2*\yr+0*\yr*\crss)[xnode,scale=\crsc]{};
            \node (cB) at (#1*\xr+0*\xr*\crss,#2*\yr-1*\yr*\crss)[xnode,scale=\crsc]{};
            
            \draw[fill=white] (cL) to (cT) to (cR) to (cB) to cycle;
            \draw[-,fill=white] (cL) to node(cLT)[midway,xnode,scale=\crsc]{}(cT) to node(cTR)[midway,xnode,scale=\crsc]{}(cR) to node(cRB)[midway,xnode,scale=\crsc]{}(cB) to node(cBL)[midway,xnode,scale=\crsc]{}(cL);
            
            \ifnum#7=1
            \ifnum#3=1
                \node (LL) at (#1*\xr-1*\xr,#2*\yr)[xnode,scale=\tscale]{};
                \draw[xedge] (cL) to (LL);
            \fi
            \ifnum#3=2
                \node (LL) at (#1*\xr-1*\xr,#2*\yr)[xnode]{};
                \draw[xedge] (cL) to (LL);
            \fi
            
            \ifnum#4=1
                \node (RR) at (#1*\xr+1*\xr,#2*\yr)[xnode,scale=\tscale]{};
                \draw[xedge] (cR) to (RR);
            \fi
            \ifnum#4=2
                \node (RR) at (#1*\xr+1*\xr,#2*\yr)[xnode]{};
                \draw[xedge] (cR) to (RR);
            \fi
            \ifnum#4=3
                \node (RR) at (#1*\xr+1*\xr-1*\xr*\crss,#2*\yr)[xnode,scale=\crsc]{};
                \draw[xedge] (cR) to (RR);
            \fi
            
            \ifnum#5=1
                \node (TTL) at (#1*\xr-\crss*\xr,#2*\yr+1*\yr)[xnode]{};
                \node (TTR) at (#1*\xr+\crss*\xr,#2*\yr+1*\yr)[xnode]{};
                \draw[xedge,#8] (TTL) to (cT) to (TTR);
            \fi
            \ifnum#5=2
                \node (TT) at (#1*\xr,#2*\yr+1*\yr-\crss*\yr)[xnode,scale=\crsc]{};
                \draw[xedge] (TT) to (cT);
            \fi
            \ifnum#5=3
                \node (TT) at (#1*\xr,#2*\yr+1*\yr)[xnode]{};
                \draw[xedge] (TT) to (cT);
            \fi
            
            \ifnum#6=1
                \node (BBL) at (#1*\xr-\crss*\xr,#2*\yr-1*\yr)[xnode]{};
                \node (BBR) at (#1*\xr+\crss*\xr,#2*\yr-1*\yr)[xnode]{};
                \draw[xedge] (BBL) to (cB) to (BBR);
            \fi
            \ifnum#6=2
                \node (BB) at (#1*\xr,#2*\yr-1*\yr)[xnode]{};
                \draw[xedge] (cB) to (BB);
            \fi
            \fi
        }
        
        \newcommand{\myclip}{\clip (-0.4*\xr,-0.4*\yr) rectangle (1.4*\xr,1.4*\yr);}
        
        \newcommand{\isgrid}{
            \Grid{A}{3}{3}{-1}{-1}{scale=0.1}{-,lightgray}{1}{1}
            \Grid{B}{1}{1}{0}{0}{scale=0.1}{xedge}{1}{1}
        }
        
        \def\xsh{2.5}
        \begin{scope}
            \node at (-0.5*\xr,1.5*\yr)[]{(a)};    
            \myclip{}
            \isgrid{}
            \cross{0}{0}{0}{2}{2}{0}{1}{};
            \node (e) at (0.5*\xr,0.1*\yr)[anchor=south west,color=red,inner sep=1pt]{$e$};
            \draw[red,thick,dotted] (e) to (cTR);
            \cross{0}{1}{0}{1}{0}{0}{1}{};
            \node at (cB)[anchor=north east,color=red,inner sep=1pt]{$a$};
            \node at (cR)[anchor=south west,color=red,inner sep=0pt]{$b$};
            
            \node (d) at (0.5*\xr,0.5*\yr)[anchor=south west,color=red,inner sep=1pt]{$d$};
            \draw[red,thick,dotted] (d) to (cRB);
            \node (h) at (1*\xr,1*\yr)[xnode,scale=\tscale]{};
            \node (hh) at (1*\xr,1*\yr+1*\yr*\yeps)[xnode,scale=\tscale]{};
            \node (hhh) at (1*\xr,1*\yr-1*\yr*\yeps)[xnode,scale=\tscale]{};
            \draw[xedge] (h) to (hh);
            \node at (hhh)[anchor=west,color=red,inner sep=1pt]{$c$};
        \end{scope}
        
        \begin{scope}[xshift=1*\xsh*\xr cm]
            \node at (-0.5*\xr,1.5*\yr)[]{(b)};    
            \myclip{}
            \isgrid{}
            \cross{0}{0}{0}{2}{2}{0}{1}{};
            \cross{0}{1}{0}{3}{0}{0}{1}{};
            \cross{1}{1}{0}{0}{0}{2}{1}{};
        \end{scope}
        
        \begin{scope}[xshift=2*\xsh*\xr cm]
            \node at (-0.5*\xr,1.5*\yr)[]{(c)};    
            \myclip{}
            \isgrid{}
            \cross{0}{0}{0}{3}{2}{0}{1}{};
            \cross{0}{1}{0}{3}{0}{0}{1}{};
            \cross{1}{1}{0}{0}{0}{0}{1}{};
            \cross{1}{0}{0}{0}{2}{0}{1}{};
        \end{scope}
        
        \begin{scope}[xshift=3*\xsh*\xr cm]
            \node at (-0.5*\xr,1.5*\yr)[]{(d)};    
            \myclip{}
            \Grid{A}{1}{3}{-1}{-1}{scale=0.1}{-,lightgray}{1}{1}
            \Grid{A}{1}{1}{1}{-1}{scale=0.1}{-,lightgray}{1}{1}
            \draw[xedge] (1*\xr,0*\yr) to (0*\xr,0*\yr);
            \draw[xedge] (0*\xr,0*\yr-1*\yr*\yeps) to (0*\xr,1*\yr) to (1*\xr+1*\xr*\crss,1*\yr);
            \draw[-,lightgray] (1*\xr+1*\xr*\crss,1*\yr) to (2*\xr,1*\yr);
            \draw[-,lightgray] (1*\xr-1*\xr*\crss,1*\yr) to (1*\xr,2*\yr) to (1*\xr+1*\xr*\crss,1*\yr);
            \cross{1}{0}{1}{0}{1}{0}{1}{blue!75!black,very thick};
            \node at (0*\xr,1*\yr)[xnode]{};
            \node at (0*\xr,0*\yr+1*\yr*\yeps)[xnode,scale=\tscale]{};
            \node at (0*\xr,0*\yr-1*\yr*\yeps)[xnode,scale=\tscale]{};
        \end{scope}
        
        \begin{scope}[xshift=4*\xsh*\xr cm]
            \node at (-0.5*\xr,1.5*\yr)[]{(e)};    
            \myclip{}
            \Grid{A}{3}{1}{-1}{-1}{scale=0.1}{-,lightgray}{1}{1}
\draw[xedge] (1*\xr,0*\yr) to (0*\xr,0*\yr);
            \draw[xedge] (0*\xr-1*\xr*\crss,1*\yr) to (1*\xr+1*\xr*\crss,1*\yr);
            \draw[-,lightgray] (-2*\xr,1*\yr) to (2*\xr,1*\yr);
            \draw[-,lightgray] (1*\xr-1*\xr*\crss,1*\yr) to (1*\xr,2*\yr) to (1*\xr+1*\xr*\crss,1*\yr);
            \draw[-,lightgray] (0*\xr-1*\xr*\crss,1*\yr) to (0*\xr,2*\yr) to (0*\xr+1*\xr*\crss,1*\yr);
            \cross{0}{0}{0}{3}{1}{0}{1}{blue!75!black,very thick};
            \cross{1}{0}{0}{0}{1}{0}{1}{blue!75!black,very thick};
        \end{scope}
        
        \begin{scope}[xshift=5*\xsh*\xr cm]
            \node at (-0.5*\xr,1.5*\yr)[]{(f)};    
            \myclip{}
            \Grid{A}{3}{1}{-1}{-1}{scale=0.1}{-,lightgray}{1}{1}
            \Grid{A}{1}{3}{-1}{-1}{scale=0.1}{-,lightgray}{1}{1}
            \draw[xedge] (1*\xr,0*\yr) to (0*\xr,0*\yr);
            \draw[xedge] (0*\xr,1*\yr) to (1*\xr+1*\xr*\crss,1*\yr);
            \draw[-,lightgray] (-2*\xr,1*\yr) to (2*\xr,1*\yr);
            \draw[-,lightgray] (1*\xr-1*\xr*\crss,1*\yr) to (1*\xr,2*\yr) to (1*\xr+1*\xr*\crss,1*\yr);
\cross{0}{0}{0}{3}{3}{0}{1}{};
            \cross{1}{0}{0}{0}{1}{0}{1}{blue!75!black,very thick};
        \end{scope}

		\end{tikzpicture}

		\caption{Grid faces created by \cref{constr:is}.}
		\label{fig:is-gf}
	\end{figure}
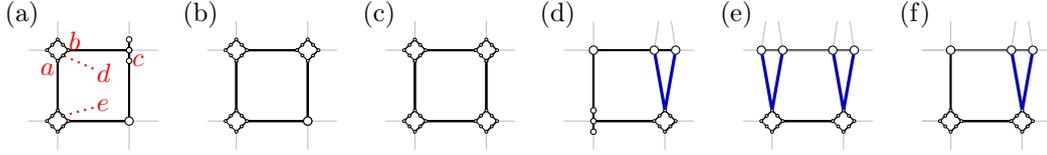

	In most cases, 
	the claim is obvious.
	We consider the face pictured in \cref{fig:is-gf}(a).
	The vertex labeled~$b$ is a GG for~$\{a,c\}$, because
	$\dist(a,c)^2 = 1$,
	$\dist(a,b)^2 = 2\eps^2$, and
	$\dist(b,c)^2 = (1-\eps)^2 + \eps^2$
	implies that for sufficiently small~$\eps$
	(i.e., $\eps\leq (\sqrt{5}-1)/4$)
	we get
	$\dist(a,c)^2 - \dist(a,b)^2 - \dist(b,c)^2 = 1 - 2 \eps - 4 \eps^2 \geq 0$.
	
	Similarly,
	$a$~is a GG blocker for~$\{d,e\}$, 
	since
	\begin{align*}
		&\dist(d,e)^2 = (1-\eps)^2, \quad
		\dist(d,a)^2 = \frac{\eps^2}{2}, \quad\text{and}\quad
		\dist(a,e)^2 = \left(1-\frac{3\eps}{2}\right)^2 + \frac{\eps^2}{4}
	\end{align*}
	implies that for sufficiently small~$\eps$
	(i.e., $\eps\leq 1/2$)
	\begin{align*}
	\dist(d,e)^2 - \dist(d,a)^2 - \dist(a,e)^2 = \eps -2\eps^2 \geq 0.
	\end{align*}
	The other cases are either obvious or follow along the same lines as these two.
	
	\emph{Relatively closest graphs.}
	In the RCG induced by the embedding described above,
	the edges highlighted in \cref{fig:is-gf}(d)--(f),
	which connect a~$\cei$ to a corner added in step~3,
	do not exist.
	This is not an issue, 
	however, 
	since the proof of the correctness of the reduction does not rely on the existence of those edges.
	Hence, 
	an adjusted reduction that omits these edges proves the \NP-hardness of the problem on RCGs.
	
	\emph{ETH-based lower bound}.
	By \cref{lemma:vc-ETH}, 
	\textsc{Vertex Cover} does not admit a~$2^{o(n^{1/2})}$-time algorithm on planar graphs with maximum degree three where $n$ is the number of vertices unless \ETHbreaks,
	and neither does~\isTsc{}.
By \cref{lemma:is}(iv), 
	the reduction only increases the number of vertices quadratically.
	Thus,
	\isAcr{} on RCGs, RNGs, and GGs
	does not admit a~$2^{o(n^{1/4})}$~time algorithm unless \ETHbreaks.
\end{proof}

\section{Dominating Set}
\label{sec:ds}
In this section, we will investigate the restriction of the following problem to RNGs, RCGs, and Gabriel graphs:

\decprob{\dsTsc{} (\dsAcr{})}{ds}
{An undirected graph~$G=(V,E)$ and~$k\in\N_0$.}
{Is there a vertex set~$X\subseteq V$ with~$|X|\leq k$ such that~$N_G[X]=V$?}

We will show:

\begin{theorem}
 \label{thm:dsNP}
 \dsTsc{} on RCGs,
 on RNGs,
 and on GGs
 is \NP-hard,
 even if the maximum degree is four,
 Moreover, unless~\ETHbreaks,
 it admits
 no~$2^{o(n^{1/4})}$-time algorithm
 where $n$ is the number of vertices.
\end{theorem}

\noindent
\dsTsc{} on $3$-regular planar graphs is claimed~\cite{Garey1979,Johnson1984}
to be \NP-hard, 
but we do not know of a published full proof.
We give a proof sketch for the following related statement:

\begin{lemma}\label{lem:dsPre}
 \dsTsc{} on planar graphs with maximum degree three
 is \NP-hard 
 and,
 unless~\ETHbreaks,
 admits no $2^{o(n^{1/2})}$-time algorithm
 where $n$ is the number of vertices.
\end{lemma}

{
\begin{proof}
	By \cref{lemma:vc-ETH}, the \textsc{Vertex Cover} problem on planar
	graphs with maximum degree three does not admit a $2^{o(n^{1/2})}$-time 
	algorithm where $n$ is the number of vertices, unless \ETHbreaks.
	\textsc{Vertex Cover} on planar graphs with maximum degree three can be 
	reduced to \dsTsc{} on planar graphs with maximum degree six using a 
	standard reduction that involves deleting isolated vertices and 
	replacing every edge with a $3$-cycle.
	This reduction only changes the size of the graph linearly.
	Hence, the same claim applies to \dsTsc{} on planar graphs with maximum degree six.
	The vertex expansion operation employed by Chen~\etal~\cite[Sect.~3.2]{Chen2009} in a similar 
	context may be used to reduce the maximum degree to three while also only changing the size of the graph linearly.
	Hence, 
	unless \ETHbreaks,
	\dsTsc{} cannot be solved in time $2^{o(n^{1/2})}$ on planar graphs with maximum degree three. where $n$ is the number of vertices.
\end{proof}
}

\noindent
We will now present the construction of a polynomial-time many-one reduction from \dsAcr{} on planar graphs with maximum degree three to \dsAcr{} restricted to graphs that are RCGs, RNGs, and GGs.

\begin{figure}[t]
	\centering
	\includegraphics{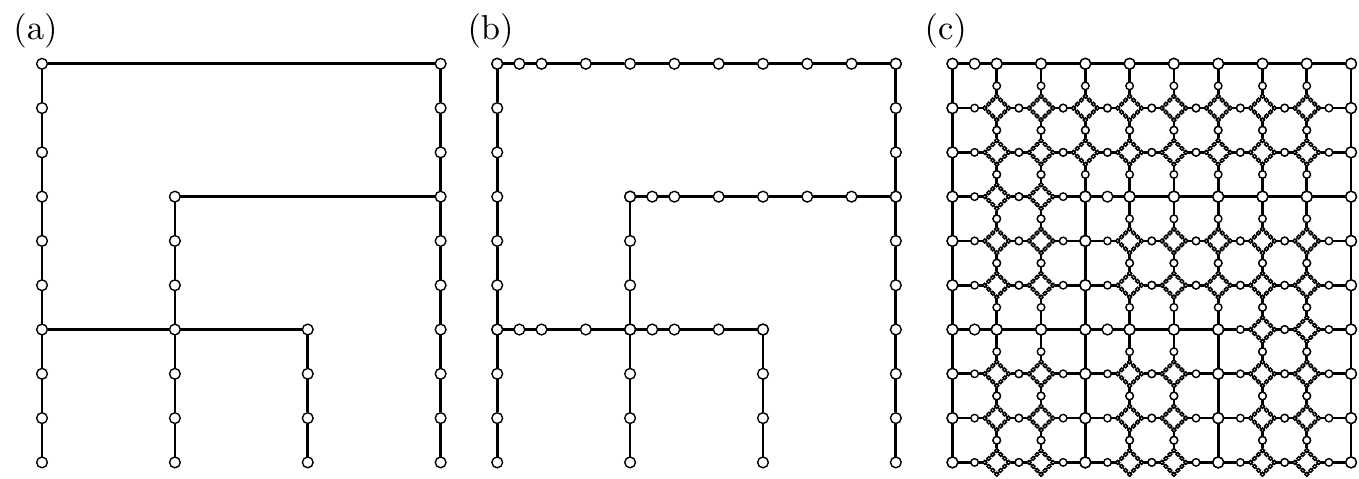}
	\caption{\cref{constr:ds} 
applied to the example graph in~\cref{fig:exgraph}.
		(a) The graph after steps 1--2.
		(b) The graph after steps 1--3.
		(c) The final graph.
	}
	\label{fig:ds}
\end{figure}

\begin{construction}
  \label{constr:ds}
	Let $G=(V,E)$ be a planar graph with maximum degree three and let~$k\in\Nzero$.
	We will construct a graph~$G'=(V',E')$ and~$k'\in\Nzero$ such that $G'$~is an RCG, 
	an RNG, 
	and a GG, 
	and $G$~contains a dominating set of size~$k$ if and only if~$G'$ contains one of size~$k'$
	(see~\cref{fig:ds} for an illustration).
	
	The vertex set~$V'$ will mostly consist of groups of vertices
	called $(x,y)$-\emph{corners}, where~$3 \leq x \leq 3n$ and $-3h_2(G) \leq y \leq 3h_1(G)$.
	Each corner consists of either a single vertex or a copy of the cycle $\ctw$.
	In the RNG-embedding of $G'$, the vertices forming the $(x,y)$-corner will be embedded roughly around the coordinate $(x,y)$.
	If $u_1,\ldots,u_{12}$ are the vertices in a copy of $\ctw$ in the order in which they appear on the cycle, then we refer to $u_1,u_4,u_7,u_{10}$ as the top, right, bottom, and left \emph{connecting vertices}.
	The connecting vertices will be the only vertices in a copy of $\ctw$ with neighbors outside of that cycle.
	If a corner consists of a single vertex, then that vertex itself simultaneously forms the top, bottom, left, and right connecting vertex of that corner.
We start with $G'\coloneqq G$ and $k' \coloneqq k$.
	
	\textbf{Step~1:} 
	Compute a 2-page book embedding of~$G$. 
	Let~$v_1,\ldots,v_n$
	be the vertices of~$G$ enumerated in the order in which they appear on the spine,
	and let~$\{E_1,E_2\}$ denote the partition of~$E$.
	
	\textbf{Step~2:} 
	Replace each vertex~$v_i$ with a path of length~$3h_1(v_i)+3h_2(v_i)$.
	The vertices on this paths form the corners~$(3i, -3h_2(v_i)),(3i, -3h_2(v_i)+1), \dots, (3i, 3h_1(v_i))$.
	Every edge~$e$ of~$G$ incident to~$v_i$ is now instead attached to the $(3i,3h(e))$-corner if~$e \in E_1$ or the $(3i,-3h(e))$-corner if~$e \in E_2$.
	Increase~$k'$ by~$h_1(v_i)+h_2(v_i)$.
	
	\textbf{Step~3:} 
	For every edge~$e=\{v_i,v_j\} \in E_1$,
	subdivide the corresponding edge in $G'$
	a total of $3\ell(e)$~times
	and increase~$k'$ by~$\ell(e)$.
	Suppose that $w_0,\ldots,w_{3\ell(e)-1}$ are the vertices introduced in the subdivision,
	where~$w_0$ is adjacent with the $(3i, 3h(e))$-corner.
	Then,
	$w_0$~is not a corner, 
	but for every~$x \in \set{3\ell(e)-1}$, 
	$w_x$ is the $(3i+x, 3h(e))$-corner.
	
	\textbf{Step~4:} 
	For every $(x,y)$ with~$3 \leq x \leq 3n$
	and~$0 \leq y \leq 3h_1(G)$, if an $(x,y)$-corner
	was not added in the previous
	two steps, then add a copy of $\ctw$ and call it the $(x,y)$-corner.
	Then add an edge connecting 
	the top connecting vertex of the $(x,y)$-corner with
	the bottom connecting vertex of the $(x,y+1)$-corner,
	the left connecting vertex of the~$(x,y)$
	corner 
	with the right connecting vertex of the $(x-1,y)$-corner, 
	and so on.
	Subdivide each of these edges once.
	
	Steps~2~to~4, 
	take only~$E_1$ into account.
	These steps must be repeated analogously for the edges in~$E_2$
	using negative $y$-coordinates.
	\cqed
\end{construction}

\begin{lemma}
	Let $G=(V,E)$ and $G'$ be graphs and let $\ds(G)$ and $\ds(G')$ denote
	the size of a smallest dominating set in each graph.
	Let $k\in \Nzero$.
	\begin{enumerate}[(i)]
		\item If $G'$ is obtained from $G$ by replacing a vertex~$v\in V$ with
		a path
consisting of $3k+1$~new vertices $v_1,\ldots,v_{3k+1}$ in that
		order and connecting each $u \in N_G(v)$ to an arbitrary vertex in~$\{v_{3i+1} \mid 0\leq i \leq k \}$,
		then $\ds(G') = \ds(G) +k$.
		\item If $G'$ is obtained from $G$ by subdividing an edge $e \in E$ a
		total of $3k$ times, then~$\ds(G') = \ds(G)+k$.
		\item If $G'$ is obtained from $G$ by adding a copy of the cycle $\ctw$
		and connecting a subset of its connecting vertices each to at most one
		existing vertex via a path of length two, then~$\ds(G') = \ds(G) +4$.
		\item If $G'$ is obtained from $G$ using \cref{constr:ds} and $n$ is
		the number of vertices in $G$, then $G'$ contains at most $\bigO(n^2)$ vertices.
	\end{enumerate}
	\label{lemma:ds}
\end{lemma}
\begin{proof}
	\begin{enumerate}[(i)]
		\item It suffices to prove the claim for $k=1$.
		The general case then follows by induction.

		Suppose that~$X\subseteq V$ is a dominating set in~$G$. If~$v \in X$,
		then~$X' \coloneqq (X \setminus \{v\}) \cup \{v_1,v_4\}$
		is a dominating set of size~$|X| + 1$ in~$G'$.
		If~$v \not \in X$, then $v$~has a neighbor~$u \in X$.
		In~$G'$,~$u$~must be adjacent to either~$v_1$ or~$v_4$.
		In the first case,~$X' \coloneqq X \cup \{v_3\}$
		is a dominating set of size~$|X| + 1$ in~$G'$.
		In the second
		case,~$X' \coloneqq X \cup \{v_2\}$ is.
		Hence,~$\ds(G') \leq \ds(G) + 1$.
		
		Now suppose that~$X'$ is a dominating set in~$G'$. First, assume
		that~$X'$ contains~$v_1$.
		Then,~$X'$ must also contain at least one of~$v_2$,~$v_3$,
		or~$v_4$ in order to dominate~$v_3$.
		Then,~$X \coloneqq (X' \setminus \{v_1,v_2,v_3,v_4\}) \cup \{v\}$
		is a dominating set in~$G$ of size at most~$|X'|-1$.
		If we assume that~$X'$ contains~$v_4$, then an analogous argument holds.
		So, assume that~$X'$ contains neither~$v_1$ nor~$v_4$.
		It must then contain one of the following:
		\begin{enumerate}[(1)]
			\item~$v_2$ and a vertex~$u \neq v_3$ that is adjacent to~$v_4$,
			\item~$v_3$ and a vertex~$u \neq v_2$ that is adjacent to~$v_1$, or
			\item both, $v_2$ and~$v_3$.
		\end{enumerate}
		In the first case, set~$X \coloneqq X' \setminus \{v_2\}$.
		In the second case, let~$X \coloneqq X' \setminus \{v_3\}$.
		In the third case, choose~$X \coloneqq (X' \setminus \{v_2, v_3\}) \cup \{v\}$.
		In each case,~$X$ is a dominating set in~$G$ of size at most~$|X'|-1$.
		Hence,~$\ds (G') \geq \ds(G) + 1$.
		\item Let~$e=\{u,v\}\in E$. 
		Subdividing~$e$ is tantamount to replacing~$u$
		with a path of length~$3k+1$, 
		while connecting~$v$ to the last vertex on this path 
		and all of~$u$'s other neighbors to the first vertex on this path. 
		Due to~(i), 
		this implies that~$\ds(G') = \ds(G) + k$.
		
		\item Suppose that $X$~is a dominating set in~$G$.
		Then, 
		$X'$ containing~$X$ and all connecting vertices
		of the copy of~$\ctw$ is a dominating set in~$G'$ and $|X'| = |X| +4$.
		Hence, 
		$\ds(G') \leq \ds(G) + 4$.
		
		Now suppose that $X'$ is a dominating set in $G'$.
		It is easy to see that $X'$ must contain at least four vertices
		in the copy of $\ctw$.
		We obtain $X$ from~$X'$ by removing all vertices of~$\ctw$
		and, if $X'$ contains any of the intermediate vertices on the
		paths from the connecting vertices to vertices in $G$, then we
		replace them by their sole neighbors in $G$.
		Then,~$X$ is a dominating set in $G$ of size at most $|X'|-4$.
		Hence,~$\ds(G') \geq \ds(G) +4$. 
		
		\item The graph $G'$ contains $\bigO (n^2)$ corners each containing $\bigO(1)$
		vertices in addition to $\bigO(n^2)$ intermediate vertices.
		\qedhere
	\end{enumerate}
\end{proof}

\noindent
We are set to prove the main result of this section.

\begin{proof}[Proof of~\cref{thm:dsNP}]
That \cref{constr:ds} is correct follows from \cref{lemma:ds}(i)-(iii).
	It is not difficult to see that
	\cref{constr:ds} may be carried out in polynomial time
	and
	outputs a graph~$G'$ with~$\O(n^2)$ vertices,
	where~$n$ denotes the number of vertices in the input graph~$G$,
	and with a maximum degree of four.
	
	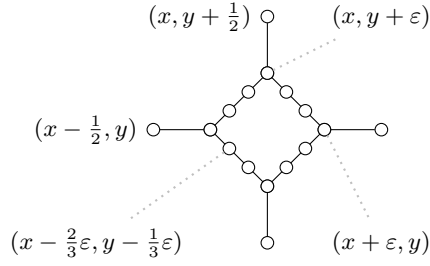
\begin{figure}[t]
		\centering
		\begin{subfigure}[t]{0.5\textwidth}
			\centering
			\begin{tikzpicture}
				\def\xr{3}
				\def\yr{3}
				\tikzpreamble{}
				
				\def\teps{0.25}
				\def\nsc{0.5}
				
				\begin{scope}
				 \draw[-] (-\teps*\xr,0) to (0,\teps*\yr) to (\teps*\xr,0) to (0,-\teps*\yr) to (-\teps*\xr,0);
				 \node (L) at (-0.5*\xr,0)[xnode,label=180:{\small$(x-\frac{1}{2},y)$}]{};
				 \node (R) at (0.5*\xr,0)[xnode]{};
				 \node (T) at (0,0.5*\xr)[xnode,label=180:{\small$(x,y+\frac{1}{2})$}]{};
				 \node (B) at (0,-0.5*\xr)[xnode]{};
				 \foreach\x in {0,1,2,3}{\node (LT\x) at (-\teps*\xr+\x*\xr*\teps/3,\x*\yr*\teps/3)[xnode]{};
          \node (LB\x) at (-\teps*\xr+\x*\xr*\teps/3,-\x*\yr*\teps/3)[xnode]{};
          \node (RB\x) at (\teps*\xr-\x*\xr*\teps/3,-\x*\yr*\teps/3)[xnode]{};
          \node (RT\x) at (\teps*\xr-\x*\xr*\teps/3,\x*\yr*\teps/3)[xnode]{};
        }
        \node (LT3l) at (0.5*\xr,0.5*\yr)[font=\small]{$(x,y+\eps)$};
        \node (RT0l) at (0.5*\xr,-0.5*\yr)[font=\small]{$(x+\eps,y)$};
        \node (LB1l) at (-0.75*\xr,-0.5*\yr)[font=\small]{$(x-\frac{2}{3}\eps,y-\frac{1}{3}\eps)$};
        \draw[thick,dotted,lightgray] (LT3) to (LT3l);
        \draw[thick,dotted,lightgray] (LB1) to (LB1l);
        \draw[thick,dotted,lightgray] (RT0) to (RT0l);
        \draw[-] (L) to (LT0);
        \draw[-] (R) to (RT0);
        \draw[-] (T) to (LT3);
        \draw[-] (B) to (LB3);
				\end{scope}

\end{tikzpicture}
		\end{subfigure}
		\caption{Embedding of a~$12$-cycle forming the $(x, y)$-corner.
		The unlabeled vertices are spread equidistantly between the vertices whose positions are given.}
		\label{fig:ds-c12-emb}
	\end{figure}
	
	It remains to show that the graph $G'$ output by \cref{constr:ds} is an RCG, 
	an RNG, 
	and a GG.
	We begin by describing an embedding of~$G'$.
	If the $(x,y)$-corner is a single vertex, then its position is $(x,y)$.
	If the $(x,y)$-corner is a copy of $\ctw$, then the vertices in the
	corner as well as the degree-2 vertices connecting the corner to 
	its surrounding corners are embedded as pictured in \cref{fig:ds-c12-emb}.
	The first vertex~$w_0$ of each path added in step~3 to replace an edge~$e=\{v_i,v_j\} \in E_1$,~$i<j$, is
	embedded at $(3i+\frac{1}{2}, 3h(e))$, and analogously for~$E_2$.
	
	\begin{figure}[t]
		\centering
		
		\begin{tikzpicture}
		 \def\xr{1}
		 \def\yr{1}
		 \tikzpreamble{}
		 \def\teps{0.25}
		 \def\nsc{0.5}
		 
		 \def\xsh{2.5}
		 
		 \newcommand{\myclip}{\clip (0.6*\xr,0.6*\yr) rectangle (2.4*\xr,2.4*\yr);}
		 
		 \newcommand{\dsgrid}{
      \Grid{G}{3}{3}{0}{0}{scale=0.1}{-,lightgray}{1}{1}
      \Grid{G}{1}{1}{1}{1}{scale=0.1}{xedge}{1}{1}
		 }
		 
		 \def\rsc{0.45}
		 \newcommand{\raute}[2]{
        \fill[white,draw=black] (-\teps*\xr+#1*\xr,0+#2*\yr) to (0+#1*\xr,\teps*\yr+#2*\yr) to (\teps*\xr+#1*\xr,0+#2*\yr) to (0+#1*\xr,-\teps*\yr+#2*\yr) to (-\teps*\xr+#1*\xr,0+#2*\yr);
				 \foreach\x in {0,1,2,3}{\node (LT\x) at (-\teps*\xr+\x*\xr*\teps/3+#1*\xr,\x*\yr*\teps/3+#2*\yr)[xnode,scale=\rsc]{};
          \node (LB\x) at (-\teps*\xr+\x*\xr*\teps/3+#1*\xr,-\x*\yr*\teps/3+#2*\yr)[xnode,scale=\rsc]{};
          \node (RB\x) at (\teps*\xr-\x*\xr*\teps/3+#1*\xr,-\x*\yr*\teps/3+#2*\yr)[xnode,scale=\rsc]{};
          \node (RT\x) at (\teps*\xr-\x*\xr*\teps/3+#1*\xr,\x*\yr*\teps/3+#2*\yr)[xnode,scale=\rsc]{};
        }
		 }
		 
		 \begin{scope}[xshift=1*\xsh*\xr cm]
		  \node at (0.5*\xr,2.5*\yr)[]{(a)};
		  \myclip{}
		  \dsgrid{}
		  \foreach[count=\i, evaluate=\i as \z using int(\i-1)] \x\y in {1/1.5,1/2,2/2,2/1,1.5/1}{
        \node (A\z) at (\x*\xr,\y*\yr)[xnode]{};
      }
      \raute{1}{1}
      \node at (A0)[anchor=east,color=red]{$d$};
      \node at (A2)[anchor=south west,color=red]{$a$};
      \node at (A4)[anchor=north,color=red]{$b$};
      \node (c) at (1.66*\xr,1.66*\yr)[inner sep=0pt,color=red]{$c$};
      \draw[thick,dotted,color=red] (c) to (1*\xr+2.35*\teps/3*\xr,1*\yr+1.35*\teps/3*\yr);
		 \end{scope}
		 
		 \begin{scope}[xshift=2*\xsh*\xr cm]
		  \node at (0.5*\xr,2.5*\yr)[]{(b)};
		  \myclip{}
		  \dsgrid{}
		  \foreach[count=\i, evaluate=\i as \z using int(\i-1)] \x\y in {1/1,1/2,1.5/2,2/2,2/1.5,1.5/1}{
        \node (A\z) at (\x*\xr,\y*\yr)[xnode]{};
      }
      \foreach\a\b in {2/1}{
        \raute{\a}{\b}
      }
		 \end{scope}
		 
		 \begin{scope}[xshift=3*\xsh*\xr cm]
		  \node at (0.5*\xr,2.5*\yr)[]{(c)};
		  \myclip{}
		  \dsgrid{}
		  \foreach[count=\i, evaluate=\i as \z using int(\i-1)] \x\y in {1/1,1/2,1.5/2,2/1.5,1.5/1}{
        \node (A\z) at (\x*\xr,\y*\yr)[xnode]{};
      }
      \foreach\a\b in {2/2,2/1}{
        \raute{\a}{\b}
      }
		 \end{scope}
		 
		 \begin{scope}[xshift=4*\xsh*\xr cm]
		  \node at (0.5*\xr,2.5*\yr)[]{(d)};
		  \myclip{}
		  \dsgrid{}
		  \foreach[count=\i, evaluate=\i as \z using int(\i-1)] \x\y in {1/1.5,1.5/2,2/1.5,2/1,1.5/1}{
        \node (A\z) at (\x*\xr,\y*\yr)[xnode]{};
      }
      \foreach\a\b in {1/1,1/2,2/2}{
        \raute{\a}{\b}
      }
		 \end{scope}
		 
		 \begin{scope}[xshift=5*\xsh*\xr cm]
		  \node at (0.5*\xr,2.5*\yr)[]{(e)};
		  \myclip{}
		  \dsgrid{}
		  \foreach[count=\i, evaluate=\i as \z using int(\i-1)] \x\y in {1/1.5,1.5/2,2/1.5,1.5/1}{
        \node (A\z) at (\x*\xr,\y*\yr)[xnode]{};
      }
      \foreach\a\b in {1/1,1/2,2/2,2/1}{
        \raute{\a}{\b}
      }
		 \end{scope}

    \end{tikzpicture}
    
\caption{The grid faces created by \cref{constr:ds}.}
		\label{fig:ds-gf}
	\end{figure}
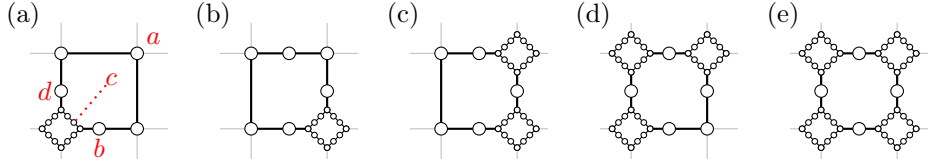
	
	We now show that the RCG, RNG, and GG induced by the vertices of any grid face
	is in fact the subgraph of~$G'$ induced by those vertices.
	For this we must show, for any pair of vertices sharing a grid face,
	that there is no RCG blocker if they are adjacent
	and that there is a GG blocker if they are not adjacent.
	We do this by examining the grid faces individually.
	\cref{fig:ds-gf} pictures all types of grid faces that may occur in $G'$ (up to symmmetry).
	
	We start with the grid pictured in \cref{fig:ds-gf}(a).
	The vertex $b$ is a GG blocker for $\{a,c\}$, since
	\begin{align*}
		\dist(a,c)^2  = \left(1-\frac{\eps}{3} \right)^2 + \left(1-\frac{2\eps}{3} \right)^2,
		\
		\dist(a,b)^2  = 1 + \frac{1}{4},
		\ \text{and}\
		\dist(b,c)^2  = \left(\frac{1}{2} - \frac{2\eps}{3} \right)^2 + \frac{\eps^2}{9}.
	\end{align*}
	Hence, if $\eps>0$ is sufficiently small 
	(i.e., $\eps<3/8$):
	\begin{align*}
		\dist(a,c)^2 - \dist(a,b)^2 - \dist(b,c)^2 = \frac{1}{2} - \frac{4\eps}{3} > 0
	\end{align*}
	It follows that $b$ is a GG blocker for $\{a,c\}$. 
	Similarly, $c$ is a GG blocker for $\{b,d\}$, 
	since
	\begin{align*}
		\dist(b,d)^2 = \frac{1}{2},
		\
		\dist(b,c)^2  = \left(\frac{1}{2} - \frac{2\eps}{3} \right)^2 + \frac{\eps^2}{9},
		\ \text{and}\
		\dist(c,d)^2  = \left(\frac{1}{2} - \frac{\eps}{3} \right)^2 + \frac{4\eps^2}{9}.
	\end{align*}
	Hence, if $\eps>0$ is sufficiently small
	(i.e., $\eps<9/10$):
	\begin{align*}
		\dist(b,d)^2 - \dist(b,c)^2 - \dist(c,d)^2 = \eps - \frac{10\eps^2}{9} > 0.
	\end{align*}
	It follows that $c$ is a GG blocker for $\{b,d\}$.
In all other cases, 
	the claim is either obvious or follows
	along the same lines as the two cases we have proved.
	
	By \cref{lemma:planar}, there can be no edges between vertices that do not share a grid face.
	Thus, we have proven that the given embedding induces~$G'$ as its RCG, RNG, and GG.
	
	\emph{ETH-lower bound.} Follows from \cref{lem:dsPre} and \cref{lemma:ds}(iv).
\end{proof}

\noindent
It remains open whether \dsAcr{} or \isAcr{} can be solved in polynomial time 
when restricted to RCGs, RNGs, or GGs with maximum degree three.

\section{Conclusion}
\label{sec:conclusion}
We have shown that problems that are NP-hard on planar graphs typically remain NP-hard on the three proximity graph classes we study.
This suggests that the main tools of algorithm theory to attack these problems 
shall be parameterized and approximation algorithms.
\isAcr{}, \dsAcr{}, and \fvsAcr{} all admit polynomial-time approximation schemes on planar graphs~\cite{Baker1994,Kleinberg2001}.
\fvsAcr{} is fixed-parameter tractable on arbitrary graphs~\cite{Bodlaender94},
while \dsAcr{} and \isAcr{} are so on planar graphs~\cite{AlberBFKN02,downey2013fundamentals}.
We are not aware of any improvements to these results that are specific to RCGs, RNGs, or GGs.

It remains an important open question whether or not RCGs, RNGs or GGs can be recognized in polynomial time
and whether an embedding for a given graph can be computed~\cite{Bose1996,Brandenburg2004,Eades1995}.
If not, then one might suspect that the graph problems we have investigated might be easier if one is given an embedding rather than just the graph.
Our reductions, however, prove that they do not, since we also give embeddings for the output graphs.

We showed that \fvsAcr{} is NP-hard on proximity graphs with maximum degree four, and it was already known to be polynomial-time solvable on any graph with maximum degree three.
For the other problems, we did not prove tight bounds on the maximum degree and open questions remain in this respect.
We proved ETH-based lower bounds of~$2^{o(n^{1/4})}$ for each of these problems on RCGs, RNGs, and GGs (with the exceptions of \tcAcr{} and \hcAcr{} on RCGs).
On general graphs,
these problems can be solved in time $2^{\bigO(n^{1/2})}$ on planar graphs and this running time is optimal unless \ETHbreaks~\cite{Lokshtanov2011}.
However, it might be possible to solve these problems on RCGs, RNGs, and GGs with a time bound
strictly between~$2^{o(n^{1/4})}$~and~$2^{\bigO(n^{1/2})}$.

More generally,
we are not aware of any problem that is know to be easier on these graph classes than on arbitrary planar graphs
(excluding trivial cases like \tcAcr{} on RCGs).
Any such example would be of interest.

\bibliography{strings-long,probs-on-rngs-bib}

\end{document}